\documentclass[journal]{IEEEtran}

\usepackage[titles]{tocloft}

\usepackage{braket}
\usepackage{amsmath,amssymb,amsfonts,dsfont,amsthm}
\usepackage{algorithmic}
\usepackage{graphicx}
\usepackage{textcomp}
\usepackage{xcolor}
\usepackage[colorlinks=true,linkcolor=black,anchorcolor=black,citecolor=black,filecolor=black,menucolor=black,runcolor=black,urlcolor=black]{hyperref}
\usepackage{orcidlink}

\usepackage[style=ieee,giveninits=false,sorting=none,minnames=3,maxnames=10,dashed=false]{biblatex}

\usepackage{lipsum}

\newcommand\myfrac[2]{\genfrac{}{}{0pt}{}{#1}{#2}}

\usepackage{flushend}

\newtheorem{theorem}{Theorem}[section]
\newtheorem{corollary}[theorem]{Corollary}
\newtheorem{lemma}[theorem]{Lemma}
\newtheorem{definition}[theorem]{Definition}
\newtheorem{proposition}[theorem]{Proposition}

\newtheorem{remark}[theorem]{Remark}

\newlength{\blank}
\newcommand{\1}{\mathds{1}}
\DeclareMathOperator{\Tr}{Tr}
\DeclareMathOperator{\id}{id}

\newcommand{\EE}{\mathbb{E}}
\newcommand{\RR}{\mathbb{R}}
\newcommand{\NN}{\mathbb{N}}
\newcommand{\ZZ}{\mathbb{Z}}

\newcommand{\cP}{\mathcal{P}}

\newcommand{\cC}{\mathcal{C}}
\newcommand{\cD}{\mathcal{D}}
\newcommand{\cE}{\mathcal{E}}
\newcommand{\cL}{\mathcal{L}}
\newcommand{\cM}{\mathcal{M}}
\newcommand{\cN}{\mathcal{N}}

\newcommand{\cS}{\mathcal{S}}
\newcommand{\cT}{\mathcal{T}}
\newcommand{\cU}{\mathcal{U}}
\newcommand{\cX}{\mathcal{X}}
\newcommand{\cY}{\mathcal{Y}}

\newcommand{\ox}{\otimes}
\newcommand{\proj}[1]{\ket{#1}\!\bra{#1}}

\begin{document}

\title{Deterministic \!identification \!over \!channels \!with \!finite \!output: \!a \!dimensional \!perspective \!on \!superlinear \!rates
\thanks{This work has been accepted for publication in IEEE Transactions on Information Theory \cite{CDBW:DI_IEEE}. A preliminary version was presented at ISIT 2024, Athens (Greece) 7-12 July 2024 \cite{CDBW:ID-continuous-discrete:ISIT}.
HB and CD are supported by the German Federal Ministry of Education and Research (BMBF) within the national initiative on 6G Communication Systems through the research hub 6G-life, grants 16KISK002 and 16KISK263, within the national initiative on Post Shannon Communication (NewCom), grants 16KIS1003K and 16KIS1005, within the national initiative ``QuaPhySI'', 
grants 16KISQ1598K and 16KIS2234, and within the national initiative ``QTOK'', 
grant 16KISQ038. 
HB has further received funding from the German Research Foundation (DFG) within Germany’s Excellence Strategy EXC-2092-390781972. 
AW is supported by the European Commission QuantERA grant ExTRaQT (Spanish MICIN project PCI2022-132965), by the Spanish MICIN (projects PID2019-107609GB-I00 and PID2022-141283NB-I00) with the support of FEDER funds, by the Spanish MICIN with funding from European Union NextGenerationEU (PRTR-C17.I1) and the Generalitat de Catalunya, and by the Alexander von Humboldt Foundation.
PC and AW are furthermore supported by the Institute for Advanced Study of the Technical University Munich (IAS-TUM).}}
\author{Pau~Colomer\textsuperscript{\orcidlink{0000-0002-0126-4521}}\thanks{P. Colomer (pau.colomer@tum.de) is with the Chair of Theoretical Information Technology, Technische Universit\"at M\"unchen (LTI-TUM), Theresienstra{\ss}e 90, 80333 M\"unchen, Germany; and IAS-TUM, Lichtenbergstra{\ss}e 2a, 85748 Garching, Germany.},~\IEEEmembership{Student Member,~IEEE}, 
Christian~Deppe\textsuperscript{\orcidlink{0000-0003-3047-3549}}\thanks{C. Deppe (christian.deppe@tu-bs.de) is with the Institute for Communications Technology and the 6G-life research hub in Technische Universit\"at Braunschweig, Schleinitzstra{\ss}e 22, 38106 Braunschweig, Germany; and was with the Institute for Communications Engineering, 
TUM.},~\IEEEmembership{Member,~IEEE},\protect\\
Holger~Boche\textsuperscript{\orcidlink{0000-0002-8375-8946}}\thanks{H. Boche (boche@tum.de) is with LTI-TUM and the 6G-life research hub in Theresienstra{\ss}e 90, 80333 M\"unchen, Germany; Munich Center for Quantum Science and Technology, Schellingstra{\ss}e 4, 80799 München, Germany; and Munich Quantum Valley, Leopoldstra{\ss}e 244, 80807 München, Germany.},~\IEEEmembership{Fellow,~IEEE},
and~Andreas~Winter\textsuperscript{\orcidlink{0000-0001-6344-4870}}
\thanks{A. Winter (andreas.winter@uab.cat) is with ICREA and Grup d'informaci\'o qu\`antica, Departament de F\'isica, Universitat Aut\`onoma de Barcelona, 08193 Bellaterra (Barcelona), Spain, as well as with IAS-TUM.}}

%
%
%


\maketitle

\begin{abstract}
Following initial work by JaJa, Ahlswede and Cai, and inspired by a recent renewed surge in interest in deterministic identification (DI) via noisy channels, we consider the problem in its generality for memoryless channels with finite output, but arbitrary input alphabets. 
Such a channel is essentially given by its output distributions as a subset in the probability simplex. 
Our main findings are that the maximum length of messages thus identifiable scales superlinearly as $R\,n\log n$ with the block length $n$,
and that the optimal rate $R$ is bounded in terms of the covering (aka Minkowski, or Kolmogorov, or entropy) dimension $d$ of a certain algebraic transformation of the output set: $\frac14 d \leq R \leq \frac12 d$. Remarkably, both the lower and upper Minkowski dimensions play a role in this result. Along the way, we present a \emph{Hypothesis Testing Lemma} 
showing that it is sufficient to ensure pairwise reliable distinguishability of the output distributions to construct a DI code.
%
%
Although we do not know the exact capacity formula, we can conclude that the DI capacity exhibits superactivation: there exist channels whose capacities individually are zero, but whose product has positive capacity. 
We also generalise these results to classical-quantum channels with finite-dimensional output quantum system, 
in particular to quantum channels on finite-dimensional quantum systems under 
the constraint that the identification code can only use tensor product inputs. 
\end{abstract}

\begin{IEEEkeywords}
Shannon theory;
identification via channels; 
quantum information.
\end{IEEEkeywords}

\renewcommand{\contentsname}{}
\setlength{\cftbeforetoctitleskip}{0pt}
\setlength{\cftaftertoctitleskip}{0pt}
\setlength{\cftbeforesecskip}{1pt}
\addtolength{\cftsecnumwidth}{1mm}
\addtolength{\cftsubsecnumwidth}{1.66mm}
\addtolength{\cftsubsecindent}{1mm}
\phantom{.}
\vspace{-1.6cm}
\tableofcontents

\thispagestyle{empty}
\section[{Communication and identification}]{Communication and identification\protect\\ via noisy channels}
\label{sec:intro}
\IEEEPARstart{I}{n} Shannon's problem of communication over a noisy channel \cite{Shannon:TheoryCommunication}, the receiver aims to faithfully recover an original message sent through a noisy memoryless channel $W:\cX\rightarrow\cY$, where $\cX$ and $\cY$ are the input and output alphabets, respectively, and which is given by the transition probabilities $W(y|x)$, for $x\in\cX$ and $y\in\cY$. That is (for discrete $\cY$), for all $x,y$, $W(y|x) \geq 0$ and for all $x$, $\sum_y W(y|x) = 1$. This means that one may consider a channel equivalently as a (measurable) function, denoted by the same letter, $W:\cX \rightarrow \cP(\cY)$, mapping inputs $x\in\cX$ to probability distributions $W_x = W(\cdot|x) \in \cP(\cY)$ in the probability simplex over $\cY$.

In block length $n\in\NN$, we have the product transition probabilities
\(
  W_{x^n}(y^n) = W^n(y^n|x^n) = \prod_{i=1}^n W(y_i|x_i) = \prod_{i=1}^n W_{x_i}(y_i)
\)
for the $n$-letter sequences (words) $x^n = x_1x_2\ldots x_n \in \cX^n$ and $y^n = y_1y_2\ldots y_n \in \cY^n$. 

\begin{definition}
\label{def:transmission_code}
An $(n,M,\lambda)$-\emph{transmission code} over $n$ uses of the memoryless channel $W$ is a family of pairs, $\left\{(u_m,\cD_m) : m\in[M]=\{1,\ldots,M\}\right\}$, with $u_m\in \cX^n$ code words over the input alphabet and $\cD_m\subset\cY^n$ pairwise disjoint subsets of the output words, $\cD_m\cap\cD_{m'}=\emptyset$ for all $m\neq m'\in[M]$, such that the error probability is bounded by $\lambda$ for all $m\in[M]$: 
\begin{equation}
  W^n(\cD_m|u_m)\geq 1-\lambda.
\end{equation}
The maximum number $M$ of messages of an $(n,M,\lambda)$-code is denoted by $M(n,\lambda)$. 
\end{definition}

The value of $\frac1n\log M$ is called the \emph{rate} of the code, and we define the \emph{capacity} of a channel $C(W)$ as the maximum rate for asymptotically faithful transmission:
\begin{equation}
  \label{eq:trans_capacity_def}
  C(W) := \inf_{\lambda>0} \liminf_{n\rightarrow\infty} \frac1n \log M(n,\lambda).
\end{equation}

\begin{theorem}[Shannon \cite{Shannon:TheoryCommunication}, Wolfowitz \cite{Wolfowitz:converse}]
\label{thm:Shannon_trans_capacity}
The transmission capacity of a memoryless channel $W$ is given by the following formula, and the strong converse holds, namely for all $\lambda\in (0;1)$, 
\[
  C(W) = \lim_{n\rightarrow\infty}\frac{1}{n}\log M(n,\lambda)
       = \max_{P\in\cP(\cX)} I(P;W).
\]
Here $\cP(\cX)$ is the set of probability distributions on $\cX$ and $I(P;W)=H(PW)-H(W|P)$ is the mutual information, using the notation
$PW = \sum_x P(x)W_x \in \cP(\cY)$, with the entropy $H(Q)=-\sum_y Q(y)\log Q(y)$ and the conditional entropy $H(W|P)=\sum_x P(x) H(W(\cdot|x))$.
\end{theorem}

Here and elsewhere in this article $\log$ and $\exp$ are to base $2$ by default, unless explicitly specified. 
In particular, the maximum number of messages that we can transmit through a noisy channel is exponential in the block length: $M(n,\lambda) \sim 2^{nR}$.
\begin{figure*}
    \centering
    \includegraphics[width=0.9\textwidth]{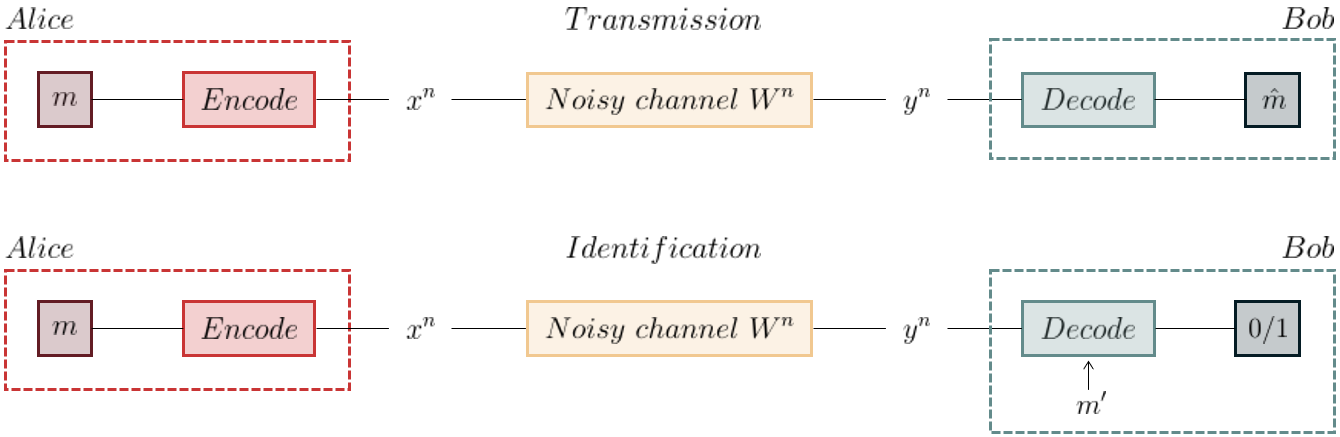}
    \caption{Alice encodes a message $m$ chosen from a set $\cM=\{1,\ldots,M\}$ into a code word of block length $n$ and sends it through $n$ uses of a noisy memoryless channel $W$. In the usual transmission scheme (above), when Bob receives $y^n$ he can decode the message, aiming to recover some $\widehat{m}\approx m$. In an identification scheme (below), he instead chooses any message $m'\in\cM$ and checks whether it is equal to $m$ with a particular hypothesis test, obtaining a binary answer.}
    \label{fig:Trans_VS_ID}
\end{figure*}

Some years after this seminal result, JaJa presented an \emph{identification} task in which instead of recovering the initial message, the receiver is only interested in knowing whether the output text corresponds to the one they have in mind, and found that this task has less communication complexity than Shannon's original transmission problem \cite{Ja:ID_easier}.

Ahlswede and Dueck fully characterized the identification problem in \cite{AD:ID_ViaChannels}, proving that identification (ID) codes can achieve doubly exponential growth of the number $N$ of messages as a function of the block length $n$, $N \sim 2^{2^{nR}}$: we can identify exponentially more messages than we can transmit. 
The main insight to prove this surprising achievability result comes from utilizing randomness in the encoder (inspired by \cite{Yao:complexity}). Starting from an $(n,M,\lambda)$-transmission code, Ahlswede and Dueck proved that reliable identification can be achieved by encoding an exponentially larger number $N \geq 2^{\lfloor{\epsilon M}\rfloor}/M$ of messages into uniform distributions over ``large'' subsets of the $M$ transmission code words. 

\begin{definition}
\label{def:ID_code}
A (randomized) $(n,N,\lambda_1,\lambda_2)$-\emph{ID code} is a family $\{(P_j,\cE_j) : j\in[N] \}$ of pairs consisting of probability distributions $P_j = P(\cdot|j)\in\cP(\cX^n)$ and subsets $\cE_j\subset \cY^n$, such that for all $j\neq k\in [N]$,
\begin{align}
  (P_jW^n)(\cE_j) &\geq 1-\lambda_1,\label{eq:I-kind}\\
  (P_jW^n)(\cE_k) &\leq \lambda_2. \label{eq:II-kind}
\end{align}
\end{definition}

Let us appreciate the formal and conceptual differences between identification and transmission codes: indeed, in randomized identification, we have probability distributions on the input, and we do not require disjointness of the output decoding sets. The latter entails the appearance of two possible error probabilities termed \emph{first} and \emph{second kind}, following the standard terminology in hypothesis testing, in contrast to transmission where we only have a single maximum error probability of incorrectly decoding $\lambda$. Here, $\lambda_1$ is the probability of a missed identification: this happens when the message Alice sends is the same as the one Bob wants to identify but, due to the noise of the protocol, the hypothesis test has a negative outcome. Likewise, $\lambda_2$ is the probability of incorrect identification: the messages sent and tested are different, but the outcome on Bob's side is positive.

Similarly to the transmission case, the maximum $N$ such that an $(n,N,\lambda_1,\lambda_2)$-ID code exists is denoted by $N(n,\lambda_1,\lambda_2)$. The asymptotic ID capacity of a channel $W$ is then defined, keeping in mind the double-exponential growth of $N$ in the block length $n$, as
\begin{equation}
  \label{eq:rand_ID_capacity_def}
  \ddot{C}_{\text{ID}}(W) := \inf_{\lambda_1,\lambda_2>0} \liminf_{n\rightarrow\infty} \frac1n \log\log N(n,\lambda_1,\lambda_2).
\end{equation}
We find it convenient to use the double dot above the capacity, 
$\ddot{C}_{\text{ID}}$, to indicate the double exponential nature of its definition. This will be useful later on as even a third capacity with yet a different scaling will be defined and compared to the others.

\begin{theorem}[{Ahlswede/Dueck~\cite{AD:ID_ViaChannels}, Han/Verd{\'u}~\cite{HanVerdu:ID}}]
\label{thm:ID_capacity_A&D}
The double exponential ID capacity of a discrete memoryless channel $W$ equals Shannon's (single) exponential transmission capacity, and 
the strong converse holds: for $\lambda_1,\lambda_2>0$, $\lambda_1+\lambda_2<1$,
\[
    \ddot{C}_{\text{ID}}(W) 
     = \lim_{n\rightarrow\infty} \frac1n \log\log N(n,\lambda_1,\lambda_2)
     = C(W).
\]
\end{theorem}
We remark here that the condition $\lambda_1+\lambda_2<1$ is there to prevent trivialities, since $\lambda_1+\lambda_2\geq 1$ can be achieved for any channel and any number of messages, by encoding them arbitrarily, and responding to all identity tests with a fixed distribution $(1-\lambda_1,\lambda_1)$.

The most general ID codes from Definition \ref{def:ID_code} utilize a randomized encoder, in the sense that if we want to send the message $j$ we make use of a probability distribution $P_j$ and therefore the particular input string that enters the channel is not deterministically defined: it could be any $x^n$ with probability $P_j(x^n)\neq 0$. 
Indeed, from the beginning, researchers were interested in what happens if we impose a deterministic encoding where, instead of probability distributions, we use code words $u_j\in\cX^n$ as in the transmission scenario. To make contact with Definition \ref{def:ID_code} for an ID code, we say that a code for identification is \emph{deterministic (DI)} if for all $j\in[N]$ we have a string $u_j\in\cX^n$ such that $P_j = \delta_{u_j}$ is the point mass concentrated on $u_j$, i.e.,
\[
  P_j(x^n) = P(x^n|j) =
  \begin{cases}
    1\quad \text{if } x^n=u_j,\\
    0\quad \text{if } x^n\neq u_j.
  \end{cases}
\]
By slight abuse of notation, we denote an $(n,N,\lambda_1,\lambda_2)$-\emph{DI code} as the family $\{(u_j,\cE_j) : j\in[N]\}$.

It was observed that this deterministic approach leads to much poorer results in terms of code size scaling in block length \cite{AD:ID_ViaChannels,AC:DI}, but without providing a complete proof. Subsequent analysis in \cite{SPBD:DI_power} finally showed that indeed DI over discrete memoryless channels (DMC) can only lead to (single) exponential scaling maximal code sizes $N(n,\lambda_1,\lambda_2)\sim 2^{nR}$ as in Shannon's paradigm, albeit with a higher rate $R$. Specifically, let $N_\text{row}(W) = |W(\cX)|$ be the number of distinct rows (output distributions) of the stochastic matrix that describes the DMC $W$. Then,
\begin{equation}
  \label{eq:Capacity_DMC}
  C_\text{DI}(W)=\log N_\text{row}(W),
\end{equation}
where $C_\text{DI}(W)=\lim_{n\rightarrow\infty}\frac{1}{n}\log N_\text{DI}(n,\lambda_1,\lambda_2)$ is the DI capacity (in the exponential scale) of a DMC. Despite this poorer performance, interest in deterministic codes has recently been renewed, as they have proven to be easier to implement and simulate \cite{DI_simpler_impl}, to explicitly construct \cite{DI_explicit_construction}, and offer reliable single-block performance \cite{AD:ID_ViaChannels}; see in particular \cite{VDTB:practical-DI} for identification codes. 

More interestingly, and directly motivating our present work, in certain channels with continuous input 
alphabets, it was found that deterministic identification codes are governed by a slightly superexponential scaling in block length. Concretely, via fast and slow fading Gaussian channels \cite{SPBD:DI_power,DI-fading} and over Poisson channels \cite{DI-poisson} optimal DI codes grow as $N(n,\lambda_1,\lambda_1) \sim 2^{Rn\log n}$. We are thus motivated to define the slightly superexponential capacity as
\begin{equation}
  \label{eq:superexp_capacity_def}
  \dot{C}_{\text{DI}}(W) 
    := \inf_{\lambda_1,\lambda_2>0} \liminf_{n\rightarrow\infty} \frac{1}{n\log n} \log N_{\text{DI}}(n,\lambda_1,\lambda_2).
\end{equation}
Indeed, the superexponential DI capacity has been bounded for fast- and slow-fading Gaussian channels $G$ \cite{SPBD:DI_power,DI-fading} (with a recently improved upper bound in \cite{VDB:DI-fading-new}) and for the Poisson channel $P$ \cite{DI-poisson,DI-poisson_mc} as follows:
\begin{equation}
  \label{eq:previousGaussPoisson}
  \frac{1}{4} \leq \dot{C}_{\text{DI}}({G})
   \leq \frac12 
   \quad \text{and} \quad 
  \frac{1}{4} \leq \dot{C}_{\text{DI}}({P}) 
   \leq \frac{3}{2}.
\end{equation}

The capacities defined until now [Equations \eqref{eq:trans_capacity_def}, \eqref{eq:rand_ID_capacity_def}, and \eqref{eq:superexp_capacity_def}], while differing in the particular scaling relevant to each case, are all given as so-called \emph{weak} and \emph{pessimistic capacities} (cf.~\cite{Ahlswede1968,Ahlswede2006} for the discussion and history of this terminology). Despite the pejorative sound of these adjectives, this is the most general and most commonly used way to define the capacity as a unique number: namely, the largest limiting rate (suitably defined) of a sequence of codes for each block length $n$, such that the error converges to zero. 
To wit, \emph{weak} (as opposed to \emph{strong}) capacity means that we take the worst (weakest) value of the rate over all positive errors (hence the infimum over $\lambda_i>0$). In other words, we make sure that the defined capacity value can be achieved with $\lambda_i\rightarrow 0$. The \emph{pessimistic} term means that the capacity is a rate realised for all sufficiently large values of $n$, i.e.~we take the lowest convergence value of the limit (the inferior limit). 
It might of course be possible to find a subsequence of block lengths $n_k \rightarrow\infty$ for which the rate converges to a higher number. This motivates the definition of an \emph{optimistic capacity}, which gives us information about the best rates that can be achieved (despite only for certain particular values of $n$) through the superior limit. In the present superexponential regime of deterministic identification this would be the following: 
\begin{equation}
  \label{eq:opt_superexp_capacity_def}
  \dot{C}_{\text{DI}}^\text{opt}(W) 
    := \inf_{\lambda_1,\lambda_2>0} \limsup_{n\rightarrow\infty} \frac{1}{n\log n} \log N_{\text{DI}}(n,\lambda_1,\lambda_2).
\end{equation}

In the present work, we study deterministic identification over a general memoryless channel, mostly with finite output but an arbitrary input alphabet. 
After reviewing preliminaries about continuous channels, typicality, distance measures, and dimension theory (Section \ref{sec:prelim}), in Section \ref{sec:hypothesis-ID} we formulate a certain \emph{Hypothesis Testing Lemma} \ref{lemma:abstract-hypo-testing} (generalising a previous insight on the intersection of typical sets \cite{Ahlswede:newmethod}) which implies that for the construction of a deterministic identification code, it is sufficient to ensure pairwise reliable distinguishability of the output distributions. 
In Section \ref{sec:sqrt-metric} we analyse the geometric implications of this observation by providing a natural metric on code words towards both the code construction and the converse. 

In Section \ref{sec:results} we prove our main results showing that slightly superexponential codes are a general feature of channels with infinite input. We will start by analysing one of the most straightforward and relevant examples, the Bernoulli channel, which given a real number $x\in[0;1]$ as input produces a binary output $y\in\{0,1\}$ according to the Bernoulli distribution $B_x$. We find the following bounds for its superexponential capacity:
\[
  \frac{1}{4} \leq\dot{C}_{\text{DI}}(B) \leq \frac12.
\]
This result yields intuitions to analyse the general case of DI through arbitrary channels with finite output. We find that the size of the code scales superexponentially with the block length $n$, and we will derive bounds for the superexponential pessimistic capacity $\dot{C}_\text{DI}(W)$ [optimistic capacity $\dot{C}_\text{DI}^\text{opt}(W)$] in terms of the lower [upper] Minkowski dimension $d$ of a certain reparametrisation of the output set $W(\cX) \subset \cP(\cY)$ (see Section \ref{sec:sqrt-metric} and Subsection \ref{subsec:main}):
\begin{equation}
  \label{eq:C_results_intr}
  \frac14 d \leq \dot{C}_{\text{DI}}(W) \leq \frac12 d.
\end{equation}
This generalises our Bernoulli example and the previous results for the Gaussian and the Poisson channels [cf.~Equation~\eqref{eq:previousGaussPoisson}], which all have $d=1$. For the Poisson channel, in Subsection \ref{subsec:poisson} we actually improve the upper bound \eqref{eq:previousGaussPoisson} on the superexponential DI capacity with a straightforward adaptation of the general method, from $3/2$ to $1/2$. 
We devote Subsection \ref{subsec:d=0} to an exploration of the slightly singular but surprisingly rich case of Minkowski dimension zero. In this case, our bounds for the capacity [Equation~\eqref{eq:C_results_intr}] tell us that we cannot achieve DI codes with $n\log n$-scaling of the message length, but do not give any further relevant information on the actual maximum code size. We overcome this issue with Theorem \ref{thm:Abstract_Generalization}, which can be seen as a further generalisation on the previous results making it possible to bound a suitably defined capacity no matter the scaling.
In Subsection \ref{subsec:superactivation}, we proceed to show that the superadditivity of the lower Minkowski dimension can be exploited to prove superactivation for DI in a classical, memoryless and i.i.d.~communication setting, which we believe to be a first. Subsection \ref{subsec:example} contains the analysis of DI over a channel with both continuous input and output, for which we can determine the slightly superexponential DI capacity exactly. 

Finally, in Section \ref{sec:q-deterministic_ID} we show that the previous classical channel results can be generalised to classical-quantum channels with finite-dimensional output quantum systems; and, in particular, to identification over quantum channels under the restriction that only tensor products are used in the encoding. We conclude in Section \ref{sec:discussion} in two parts: on the one hand an extended reflection on the occurrence of optimistic and pessimistic DI capacities in our work, and their bearing on the observed superactivation. On the other hand, with a discussion of what we have achieved in the present work, giving us occasion to highlight several open questions.

\section{Preliminaries}
\label{sec:prelim}
Consider a channel $W:\mathcal{X}\rightarrow\mathcal{Y}$, where the output alphabet $\mathcal{Y}$ is a finite set, the case of principal interest in the present paper, though most of what we do extends to countably infinite $\mathcal{Y}$, and we shall discuss continuous measure spaces at a few points. The input alphabet $\mathcal{X}$ is an arbitrary measurable space. This means really that we have a measurable map $W:\mathcal{X}\rightarrow\cP(\mathcal{Y})$ from the input space to the probability simplex over $\mathcal{Y}$, the latter equipped with the Borel $\sigma$-algebra: every input $x\in\mathcal{X}$ is mapped to a probability distribution $W_x$ on $\mathcal{Y}$, which we identify with its probability vector of $|\mathcal{Y}|$ components. 

On the face of it, the measure theory of $\mathcal{X}$ might make the channel arbitrarily complicated. However, noting that the output distributions $W_x\in\cP(\mathcal{Y})$ determine everything about the channel's capacities, 
we can straight away move to a standardised version of the channel, by identifying $\mathcal{X}$ with the image of the channel, $\widetilde{\mathcal{X}} := W(\mathcal{X}) \subset \cP(\mathcal{Y})$. This leads to an equivalent channel with input alphabet $\widetilde{\mathcal{X}}$ and acting as the embedding map $\widetilde{\mathcal{X}} \subset \cP(\mathcal{Y})$. 
For the purposes of communication and identification in the Shannon setting of non-zero errors ($\lambda_1,\lambda_2>0$), we may then w.l.o.g.~assume that $\widetilde{\mathcal{X}}$ is closed, because otherwise we can pass to the closure which has the same asymptotic rate and error characteristics.

We shall require basic tools from typicality. Let us define, for block length $n$ and a point $x^n\in\cX^n$, the \emph{(entropy) conditional typical set} in $\cY^n$ as
\begin{equation}
  \label{eq:entropy-typical-set}
  \mathcal{T}_{x^n}^\delta \! := \left\{ y^n\in\mathcal{Y}^n \!: \left| \log W_{x^n}(y^n) \!+\! H(W_{x^n}) \right| \leq \delta\sqrt{n} \right\}.
\end{equation}
For sufficiently large $n$, the conditional typical set satisfies the following properties \cite[Lemmas~I.11~and~I.12]{winter:PhDThesis}:
\begin{enumerate}
    \item \emph{Unit probability}: the set $\cT_{x^n}^\delta$ asymptotically has probability $1$. Quantitatively, 
    \begin{equation}
      \label{eq:entropy-typical-prob}
        W_{x^n}(\mathcal{T}_{x^n}^\delta) 
         \geq 1 - \frac{K(|\mathcal{Y}|)}{\delta^2},
    \end{equation}
    where $K(d) = \left(\log\max\{d,3\}\right)^2$.
    
    \item \emph{Equipartition}: the probability of all conditionally typical sequences is approximately uniform. To be precise, for all $y^n\in\mathcal{T}_{x^n}^\delta$,
    \begin{equation}\begin{split}
      \label{eq:entropy-typical-pointwise}
        W_{x^n}(y^n)&\leq 2^{-H(W_{x^n})+\delta\sqrt{n}}, \\
        W_{x^n}(y^n)&\geq2^{-H(W_{x^n})-\delta\sqrt{n}}.
   \end{split}\end{equation}
\end{enumerate}

While the error bound \eqref{eq:entropy-typical-prob} would do for our rate proofs, where we may fix $\delta$ sufficiently large or let it grow very slowly with $n$, the following lemma gives (almost-)exponentially small bounds.

\begin{lemma}
\label{lemma:entropy-typical-prob}
For an arbitrary channel $W:\mathcal{X}\rightarrow\mathcal{Y}$, arbitrary block length $n$, $0<\delta\leq \sqrt{n}\log|\cY|$, and for any $x^n\in\mathcal{X}^n$, we have
\[
  W_{x^n}(\mathcal{T}_{x^n}^\delta) 
   \geq 1-2\exp\left(-\delta^2/36 K(|\mathcal{Y}|)\right),
\]
where $K(d) = (\log\max\{d,3\})^2$ as before.
\end{lemma}
\begin{proof}
Consider the random variables $Y^n=Y_1\ldots Y_n \sim W_{x^n}$ defined through the joint distribution of the channel output, and define $L_i := -\log W_{x_i}(Y_i)$, so that $-\log W_{x^n}(Y^n) = \sum_{i=1}^n L_i$. The $L_i$ are evidently independent random variables with $\EE L_i = H(W_{x_i})$ for all $i$. Since 
\[
Y^n \in \cT_{x^n}^\delta 
   \ \text{ iff }\ 
  H(W_{x^n})-\delta\sqrt{n} \leq \sum_{i=1}^n L_i \leq H(W_{x^n})+\delta\sqrt{n},
\]
we will be done once we prove
\begin{equation*}\begin{split}
  \label{eq:large-deviation}
  \Pr\left\{ \sum_{i=1}^n L_i > \sum_{i=1}^n H(W_{x_i}) + \delta\sqrt{n} \right\} 
   &\leq \exp\left(\frac{-\delta^2}{36 K(|\mathcal{Y}|)}\right), \\
  \Pr\left\{ \sum_{i=1}^n L_i < \sum_{i=1}^n H(W_{x_i}) - \delta\sqrt{n} \right\} 
   &\leq \exp\left(\frac{-\delta^2}{36 K(|\mathcal{Y}|)}\right).  
\end{split}\end{equation*}
To this end, we use the well-known Bernstein trick (cf.~\cite{DemboZeitouni}), which relies on calculating the moment generating function: for $|\lambda| < 1$, it is given by
\[\begin{split}
\EE \exp(\lambda L_i) 
= \sum_y W_{x_i}(y) W_{x_i}(y)^{-\lambda}
&= \sum_y W_{x_i}(y)^{1-\lambda}\\
&= \exp\Bigl(\lambda H_{1-\lambda}(W_{x_i})\Bigr),
\end{split}
\]
with $H_\alpha(Q) = \frac{1}{1-\alpha} \log\left(\sum_y Q(y)^\alpha\right)$ the R\'enyi entropy of order $\alpha$. By Markov's inequality, and using the shorthand $\tau = \delta/\sqrt{n}$, this gives us 
\begin{align}\label{eq:bernstein-upper}
\forall \lambda>0\quad
 \Pr&\left\{ \sum_{i=1}^n L_i > \sum_{i=1}^n H(W_{x_i}) + n\tau \right\}\\
&\nonumber \leq \exp\left( \lambda\sum_{i=1}^n \bigl(H_{1-\lambda}(W_{x_i})-H(W_{x_i})-\tau\bigr) \right),\\
\label{eq:bernstein-lower}
\forall \lambda<0\quad
 \Pr&\left\{ \sum_{i=1}^n L_i < \sum_{i=1}^n H(W_{x_i}) - n\tau \right\}\\
&\nonumber \leq \exp\left( \lambda\sum_{i=1}^n \bigl(H_{1-\lambda}(W_{x_i})-H(W_{x_i})+\tau\bigr) \right),
\end{align}
where we have used the independence of the different $\exp(\lambda L_i)$ to evaluate the expectation of their product. Since $H_{1-\lambda}(Q) \rightarrow H(Q)$ uniformly in $Q\in\cP(\cY)$ as $\lambda\rightarrow 0$, we get $H_{1-\lambda}(Q)-H(Q) \leq \frac{\tau}{2}$ ($\lambda>0$) and $H_{1-\lambda}(Q)-H(Q) \geq -\frac{\tau}{2}$ ($\lambda<0$) for sufficiently small $|\lambda|$, with a bound only depending on $|\cY|$ and $\tau$. This already is sufficient to see that we get some form of exponential bound. 

To get the explicit form claimed above, we invoke \cite[{Lemma~8}]{TCR:fully-q-AEP}, which in the simplified form that we need states that given $\Upsilon = 1+2\sqrt{|\cY|}$, 
\begin{equation}
  \label{eq:TCR-alpha}
H(Q) \geq H_\alpha(Q) \geq H(Q) - 4(\alpha-1)(\log\Upsilon)^2,
\end{equation}
for all $1\leq\alpha\leq 1+\log 3/(4\log\Upsilon)$; and
\begin{equation}
  \label{eq:TCR-beta}
H(Q) \leq H_\beta(Q) \leq H(Q) + 4(1-\beta)(\log\Upsilon)^2 
\end{equation}
for all $\forall 1\geq\beta\geq 1-\log 3/(4\log\Upsilon)$. This allows us choose $\lambda = \pm \frac{\tau}{8(\log\Upsilon)^2}$ in Equations \eqref{eq:bernstein-upper} and \eqref{eq:bernstein-lower}, respectively, because then $|\lambda| = |\alpha-1| = |1-\beta| \leq \log 3/(4\log\Upsilon)$ by our assumption $0<\tau\leq\log|\cY|$. This means that the entropy differences $H_{1-\lambda}(W_{x_i})-H(W_{x_i})$ are bounded by $\pm\frac{\tau}{2}$ and we get the following upper bound on the exponential probability bounds: 
\[\begin{split}
  \text{RHS of \eqref{eq:bernstein-upper},\, \eqref{eq:bernstein-lower}}
  &\leq \exp\left(-n\tau^2/16(\log\Upsilon)^2\right)\\
  &\leq \exp\left(-\delta^2/36K(|\cY|)\right),
\end{split}\]
where we have recalled $\delta=\tau\sqrt{n}$, and used $\log\Upsilon \leq \log 3 + \frac12\log|\cY| \leq \frac32\log\max\{3,|\cY|\}$.
\end{proof}

For our code constructions and converses, we need the packing and covering of general sets in arbitrary dimensions. The box (or sphere) \emph{packing} problem consists of finding the maximum number of disjoint hypercubes (or hyperspheres) of given side (radius), each centered within the given set. Similarly, the \emph{covering} problem consists of finding the minimum number of hypercubes (or hyperspheres) that can cover the whole set. These packing and covering numbers are fundamental in geometry, especially in the theory of fractals, as they can be used to define the dimension of a subset in Euclidean space. The shape of the basic set used to pack or cover is not important, as long as it is bounded and has non-empty interior. 

For concreteness, for a non-empty bounded subset $F$, denote $\Gamma_\delta(F)$ the minimum number of closed balls of radius $\delta$ centered at points in $F$ such that their union contains $F$ (covering), and $\Pi_\delta(F)$ the maximum number of pairwise disjoint open balls of radius $\delta$ centered at points in $F$ (packing). Note that the centers of the balls in either case form a subset $F_0 \subset F$. For a covering, $F_0$ has to be such that for every $x\in F$ there exists $x_0\in F_0$ with $d(x,x_0)\leq\delta$, which is otherwise known as $\delta$-net. For a packing, the requirement is that for any $x_0\neq x_1\in F_0$, $d(x_0,x_1) \geq 2\delta$. A fundamental observation is that for every $\delta>0$ and $\eta>0$,
\begin{equation}
  \label{eq:covering-vs-packing}
  \Pi_{\delta+\eta}(F) 
    \leq \Gamma_\delta(F) 
    \leq \Pi_{\delta/2}(F). 
\end{equation}
To see the left-hand inequality, note that the centers $F_0$ of any $(\delta+\eta)$-packing have pairwise distance $\geq 2\delta+2\eta > 2\delta$, hence each ball of any $\delta$-covering can contain at most one element from $F_0$. To see the right-hand inequality, consider any maximal $\delta/2$-packing with centers $F_0$ (i.e.,~one that cannot be increased by adding more points); by doubling the radii we necessarily obtain a covering (of the same cardinality), for if it were not a covering, this would mean that there exists an $x\in F$ at distance larger than $2\delta$ from every $x_0\in F_0$, contradicting the maximality of the packing.

With these, we can define the \emph{Minkowski dimension} (also known as covering dimension, Kolmogorov dimension, or entropy dimension, see \cite[Chapter 3]{Falconer:fractal}) as
\begin{equation}\label{eq:Minkowski_Dimension}
d_M(F) = \lim_{\delta\rightarrow0} \frac{\log \Gamma_\delta(F)}{-\log\delta} 
  = \lim_{\delta\rightarrow0} \frac{\log \Pi_\delta(F)}{-\log\delta}.
\end{equation}

The Minkowski dimension captures geometric complexity beyond the scope of the classical topological dimension. Indeed, for smooth manifolds the topological and the Minkowski dimensions coincide. More generally, an open ball, or any set that contains an open ball in $\RR^D$, have Minkowski dimension $D$, and that is the maximum for subsets of $\RR^D$. Additionally, any smooth $d$-dimensional manifold smoothly embedded or immersed in $\RR^D$ will have Minkowski dimension $d$. However, a disparity between topological and Minkowski dimensions arises from the absence of smoothness in fractals and irregular structures. A good example is the Koch snowflake (Figure \ref{fig:snowflake}), which has topological dimension 1 (as it is a curve), but its Minkowski dimension is $\log_3 4 > 1$, reflecting the capacity of the curve to occupy space at finer scales. Also, the celebrated Weierstrass function, which is continuous but nowhere differentiable, has a graph that, despite being the continuous image of an interval on the real line, has Minkowski dimension strictly between $1$ and $2$.

\begin{figure}[ht]
    \centering
    \includegraphics[width=\linewidth]{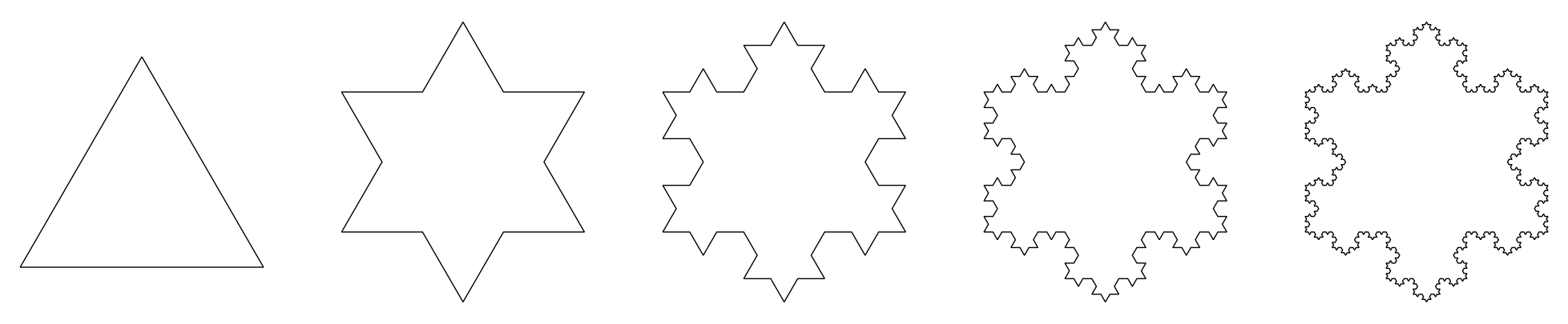}
    \caption{Iterative construction of the Koch fractal. In each iteration, a triangle bend is added to each side of the current iteration, all of which are hence polygons, though the limit is not.}
    \label{fig:snowflake}
\end{figure}

Most well known fractals and the Weierstrass example are self-similar across all different scales, which facilitates the existence of the limit in Equation \eqref{eq:Minkowski_Dimension}. However, for more general and irregular cases, it is possible (and quite common) that the above limit does not exist. In that case, we define the \emph{upper} and \emph{lower Minkowski dimensions} through the limit superior and limit inferior, respectively: 
\begin{align}
 \overline{d}_M(F) &:= \limsup_{\delta\rightarrow0}\frac{\log \Gamma_\delta(F)}{-\log\delta} 
  = \limsup_{\delta\rightarrow0} \frac{\log \Pi_\delta(F)}{-\log\delta}, \\
 \underline{d}_M(F) &:= \liminf_{\delta\rightarrow0}\frac{\log \Gamma_\delta(F)}{-\log\delta}
  = \liminf_{\delta\rightarrow0} \frac{\log \Pi_\delta(F)}{-\log\delta}.
\end{align}
That the limits are the same whether we use covering ($\Gamma$) or packing ($\Pi$) follows from the chain of inequalities \eqref{eq:covering-vs-packing}. We further remark that, as a finite covering of $F$ is automatically a covering of its closure $\overline{F}$, the (upper and lower) Minkowski dimensions remain invariant when passing from $F$ to $\overline{F}$. For further reading on dimension theory of general sets and fractal geometry, we refer the reader to the excellent textbook \cite{Falconer:fractal}, as well as the texts \cite{Robinson:dimensions,Fraser:dimensions}.

Throughout the proofs and discussion, we will also need some distance measures between probability distributions which are defined next. The \emph{total variation distance} is a statistical distance measure which coincides with half the $L^1$ norm between the probability density functions. Let $P$ and $Q$ be two probability functions defined on a finite or countably infinite measurable space $\cL$, then the total variation distance is given by 
\[
  \frac{1}{2}\|P-Q\|_1 
    = \sum_{\ell\in\cL}\frac{1}{2} \lvert{P(\ell)-Q(\ell)}\rvert.
\]
The \emph{Bhattacharyya coefficient} (in quantum information called \emph{fidelity}) is given by $F(P,Q)=\sum_{\ell\in\cL}\sqrt{P(\ell)Q(\ell)}$, and it is related to the total variation distance by the following bounds:
\begin{equation}
  \label{eq:Classical_FvdG}
  1-F(P,Q) \leq \frac12 \|P-Q\|_1 \leq \sqrt{1-F(P,Q)^2}.
\end{equation}

\section{Hypothesis testing and identification}
\label{sec:hypothesis-ID}
The outputs of reliable identification codes (deterministic or randomised) necessarily form a set of probability distributions that are pairwise well distinguishable. Indeed, in a general ID code for a memoryless channel according to Definition \ref{def:ID_code}, the concatenation of the encoding $P_j$ and the channel $W^n$ is a probability distribution $P_jW^n\in\cP(\cY^n)$. Now, the code requires the error bounds of first and second kind in Equations \eqref{eq:I-kind} and \eqref{eq:II-kind}, which directly imply that the hypothesis tests $\{\cE_j,\cY^n\setminus\cE_j\}$ defined by the decoding sets of a good ID code can distinguish the distributions on the output reliably: $\frac12\|P_jW^n-P_kW^n\|_1 \geq 1-\lambda_1-\lambda_2$. 
In other words, the output distributions $P_jW^n$ of an identification code have to form a packing in the probability simplex with respect to the total variation metric.

In the present section we analyse to which extent the converse holds: does pairwise distinguishability on the output distributions of a channel imply the existence of a good identification code?
It is not hard to come up with examples showing that in general, the answer to this question is no. In fact, Watanabe \cite{Watanabe:minmax} has given a much more complicated necessary and sufficient hypothesis testing criterion for ID codes, involving the convex hulls of all but one message. However, adding the condition that these output distributions are product distributions (as is the case for deterministic codes) and have similar entropy, turns out to be enough. 

We start by showing a general form for binary hypothesis testing of product distributions that directly applies to deterministic identification. Namely, we prove that typical sets are suitable for binary hypothesis testing. This reduces the identification problem to finding output probability distributions that are sufficiently separated in total variation distance, under the condition that the inputs that generated those distributions on the output have similar entropy. Indeed, Theorem \ref{thm:ID-from-packing} below describes a deterministic identification code that can be built directly from a sufficiently distant packing in the output, simplifying the analysis a great deal.

The following general lemma draws inspiration from \cite[Appendix]{Ahlswede:newmethod}, in particular Lemma I${}_1$, which gives a bound on the intersection of conditionally typical sets. In Ahlswede's work, only finite input alphabets $\mathcal{X}$ are considered and the typical sets are implemented by way of strong conditional typicality. We lift the discrete restriction and relax the typical sets to entropy typicality. 

\begin{lemma}
\label{lemma:abstract-hypo-testing}
Given $0<\delta\leq\sqrt{n}\log|\cY|$, consider two points $x^n,{x'}^n \in \mathcal{X}^n$ such that $1-\frac12\left\|W_{x^n}-W_{{x'}^n}\right\|_1 \leq \epsilon$. Then,
\[\begin{split}
  W_{{x'}^n}(\mathcal{T}_{x^n}^\delta) 
   &\leq 2\exp\left(-\delta^2/36K(|\mathcal{Y}|)\right)\\
   &\quad+ \epsilon\left(1 + 2^{2\delta\sqrt{n}}2^{H(W_{x^n})-H(W_{{x'}^n})}\right).
\end{split}\]
\end{lemma}
\begin{proof}
The trace distance condition means that there exists a subset (test) $\mathcal{S}\subset\mathcal{Y}^n$ that well distinguishes the two distributions $W_{x^n}$ and $W_{{x'}^n}$: 
\begin{equation}
  \label{eq:optimal-hypthesis-test}
  W_{x^n}(\mathcal{S}) \geq 1-\epsilon, \quad
  W_{{x'}^n}(\mathcal{S}) \leq \epsilon.
\end{equation}
We now observe 
\[\mathcal{T}_{x^n}^\delta \subset \left((\mathcal{T}_{x^n}^\delta\cap\mathcal{T}_{{x'}^n}^\delta)\setminus\mathcal{S}\right) \cup \mathcal{S} \cup \left(\mathcal{Y}^n\setminus\mathcal{T}_{{x'}^n}^\delta\right),
\]
so via the union bound, we can upper-bound the probability in question by the sum of three terms. Indeed, by Equations \eqref{eq:optimal-hypthesis-test} and \eqref{eq:entropy-typical-prob}, or better still Lemma \ref{lemma:entropy-typical-prob}, we have 
\[
W_{{x'}^n}(\mathcal{S}) \leq \epsilon, \  \
W_{{x'}^n}\left(\mathcal{Y}^n\setminus\mathcal{T}_{{x'}^n}^\delta\right) 
\leq 2\exp\left(-\delta^2/36K(|\mathcal{Y}|)\right)
\]
for the second and third probability, respectively. For the first term, which is the nontrivial one, we take from Equation \eqref{eq:optimal-hypthesis-test} 
\[
  W_{x^n}(\mathcal{T}_{x^n}^\delta\setminus\mathcal{S}) 
   \leq W_{x^n}(\mathcal{Y}^n\setminus\mathcal{S}) 
   \leq \epsilon,
\]
hence using the lower estimate in Equation \eqref{eq:entropy-typical-pointwise} we arrive at the cardinality bound
\[
  |\mathcal{T}_{x^n}^\delta\setminus\mathcal{S}|
   \leq \epsilon\cdot 2^{H(W_{x^n})+\delta\sqrt{n}}.
\]
Finally, applying \eqref{eq:entropy-typical-pointwise} once again, but for $W_{{x'}^n}$ and using the upper estimate, we conclude
\[
  W_{{x'}^n}\left((\mathcal{T}_{x^n}^\delta\cap\mathcal{T}_{{x'}^n}^\delta)\setminus\mathcal{S}\right)
   \leq \epsilon\cdot 2^{2\delta\sqrt{n}} 2^{H(W_{x^n})-H(W_{{x'}^n})}.
\]
Putting the three terms together finishes the proof.
\end{proof}

In general, for an identification code with errors $\lambda_1$ and $\lambda_2$, the output distributions of two code words $u,v\in\mathcal{X}^n$ necessarily have to have ``large'' total variation distance: $\frac12\|W_u-W_v\|_1 \geq 1-\lambda_1-\lambda_2$. What this means is that the $W_u$ of the code have to form a good packing in the probability simplex with respect to the total variation metric.
However, on its own, that is not sufficient, because a large pairwise total variation distance just means that there is a test well distinguishing any given pair, in other words, $\mathcal{S}$ depends on two code words rather than one. The above result means that if the outputs $W_u$ in addition have roughly equal entropy, then the entropy-typical set for each code word is a decent surrogate of the optimal test. Crucially, it therefore only depends on the one code word and is hence universal to test against all other possible code words. As a corollary, we obtain the following construction method for DI codes for $W$:

\begin{theorem}
\label{thm:ID-from-packing}
Let a memoryless channel $W:\mathcal{X}\rightarrow\mathcal{Y}$, block length $n$ and $0<\delta\leq\sqrt{n}\log|\cY|$ be given, and assume $N$ points $u_j\in\mathcal{X}^n$ ($j\in[N]$) with the property
\(
  1-\frac12\|W_{u_j}-W_{u_k}\|_1 \leq 2^{-3\delta\sqrt{n}}
\), 
for all $j\neq k$, i.e.~the $W_{u_j}$ form a packing in $\cP(\cY^n)$. Then there is a subset $\mathcal{C}\subset[N]$ of cardinality $|\mathcal{C}| \geq N/\left\lceil n\log|\mathcal{Y}|\right\rceil$ such that the collection
\(
  \left\{\left(u_j,\cE_j=\mathcal{T}_{u_j}^\delta\right) : j\in\mathcal{C}\right\}
\)
is an $(n,|\cC|,\lambda_1,\lambda_2)$-DI code, with 
\[\begin{split}
  \lambda_1 &= 2\exp\left(-\delta^2/36K(|\mathcal{Y}|)\right),
  \quad\\
  \lambda_2 &= 2\exp\left(-\delta^2/36K(|\mathcal{Y}|)\right) + 3\exp\left(-\delta\sqrt{n}\right).
\end{split}\]
\end{theorem}
\begin{proof}
We start by dividing $[N]$ into $S=\left\lceil n\log|\mathcal{Y}|\right\rceil$ parts $\mathcal{C}_s$ ($s=1,\ldots,S$) in such a way that all $j\in\mathcal{C}_s$ have $H(W_{u_j})\in[s-1;s]$. This is possible since the entropies of the distributions $W_{x^n}$ are in the interval $[0;n\log|\mathcal{Y}|]$.

Then, choose $\mathcal{C}$ as the largest $\mathcal{C}_s$, which satisfies the cardinality lower bound by the pigeonhole principle. Finally, we note that within $\mathcal{C}$, the entropies $H(W_{u_j})$ differ by at most $1$ from each other, so we can apply Lemmas \ref{lemma:entropy-typical-prob} and \ref{lemma:abstract-hypo-testing} to bound the error probabilities of first and second kind as claimed.
\end{proof}

Since we need to achieve a packing with pairwise total variation distance being almost exponentially close to $1$, and because of 
\begin{equation*}
  \label{eq:fidelity-bound}
  1-\frac12\|W_{x^n}-W_{{x'}^n}\|_1
   \leq F\left(W_{x^n},W_{{x'}^n}\right)
   =    \prod_{i=1}^n F\left(W_{x_i},W_{x_i'}\right),
\end{equation*}
where we have invoked the classical Fuchs-van de Graaf relations \eqref{eq:Classical_FvdG} first and the multiplicativity of the fidelity under tensor products after, we are motivated to study the negative logarithm of this difference: 
\begin{equation}\begin{split}
  \label{eq:log-fidelity-bound}
  -\ln\!\left(\! 1-\frac12\|W_{x^n}-W_{{x'}^n}\|_1 \!\right) 
   &\geq \sum_{i=1}^n \!-\ln F\left(W_{x_i},W_{x_i'}\right) \\
   &= \frac12 \sum_{i=1}^n -\ln F\left(W_{x_i},W_{x_i'}\right)^2 \\
   &\geq \frac12 \sum_{i=1}^n \!\left[1-F\left(W_{x_i},W_{x_i'}\right)^2\right]\!,
\end{split}
\end{equation}
and using again the classical Fuchs-van de Graaf relations we get
\[
  -\ln\!\left(\! 1-\frac12\|W_{x^n}-W_{{x'}^n}\|_1 \!\right) \geq \frac12 \sum_{i=1}^n \left(\frac12\|W_{x_i}-W_{x_i'}\|_1\!\right)^2.
\]
Notice now that this equals the square of a norm on $\left(\mathbb{R}^{\cY}\right)^n$, defined as 
\[
  \frac12 \left\| \bigoplus_{i=1}^n W_{x_i} - \bigoplus_{i=1}^n W_{x_i'} \right\|_{1,2} 
   := \sqrt{\sum_{i=1}^n \left(\frac12\|W_{x_i}-W_{x_i'}\|_1\right)^2},
\]
which is an instance of a \emph{mixed $(p,q)$-norm}, in this case a $2$-norm of a vector of (Schatten) $1$-norms.
Thus, for our purposes it will be enough that we find a $\sqrt{6\delta}\sqrt[4]{n}$-packing of points $\bigoplus_{i=1}^n W_{x_i}$ in $\mathcal{P}(\mathcal{Y})^n$, with respect to the metric induced by this norm. 

Actually, any norm defined on the single-letter output probability simplex would work equally as all norms on $\RR^{\cY}$ are equivalent, so changing the trace norm for another norm would only result in a universal constant prefactor that does not depend on the block length. This means in particular that we could perfectly well use an $f\sqrt{\delta}\sqrt[4]{n}$-packing in the metric induced by the overall Euclidean norm on the points $\bigoplus_{i=1}^n W_{x_i}$, with $f$ some constant depending on the initial norm.

\section[{Spherisation and metric on product distributions}]{Spherisation and the natural metric\protect\\ on product distributions}
\label{sec:sqrt-metric}
It is worth pausing at this juncture to get a deeper appreciation of the geometric requirements, both necessary and sufficient, of a DI code. We have early on reflected on the fact that any general ID code necessarily is a packing in $\cP(\mathcal{Y}^n)$ with respect to the total variation distance, and at the end of the previous section we have introduced a metric on $\cP(\mathcal{Y})^n$ to formulate a sufficient criterion for a packing to yield a DI code. Let us introduce some more concepts to connect these two strands into a unified and natural metric on product distributions. 

For a finite or countably infinite set $\cL$, and two probability distributions $P,Q \in \cP(\cL)$, define the \emph{purified distance} $p(P,Q) := \sqrt{1-F(P,Q)^2}$, where $F(P,Q)$ the fidelity (Bhattacharyya coefficient) introduced earlier. This is indeed a metric, its name coming from the more general version for quantum states \cite{TCR:duality}. 

For a probability distribution $P\in\cP(\cL)$, define the unit vector $\sqrt{P}:=(\sqrt{P(\ell)}:\ell\in\cL) \in \RR^{\cL}$. The image of $\cP(\cL)$ under the square root map is the non-negative orthant sector of the unit hypersphere in $\RR^{\cL}$, which we shall denote $S_+(\cL,1)$. For a channel $W:\cX \rightarrow \cY$, the image of $\widetilde{\cX}$ we denote 
\[
  \sqrt{\!\widetilde{\cX}} = \{\sqrt{W_x} : x\in\cX \} 
  \subset S_+(\cY,1).
\]
Note that the square root map is a homeomorphism between $\cP(\cY)$ and $S_+(\cY,1)$, but with respect to the $L^1$-norm on the former and the $L^2$-norm on the latter it is not Lipschitz continuous (though its inverse is Lipschitz), which means that the Minkowski dimension of $\sqrt{\!\widetilde{\cX}}$ is larger than or equal to that of $\widetilde{\cX}$. However, the map is $\frac12$-H\"older, i.e., the distance of image points is bounded by the square root of the distance between the pre-images, implying that in general
\begin{equation}
  \label{eq:dimension-comparison}
  d_M(\widetilde{\cX}) 
   \leq d_M\!\left(\!\sqrt{\!\widetilde{\cX}}\right) 
   \leq 2d_M(\widetilde{\cX}),
\end{equation}
where $d_M$ can also be replaced in all three terms by the lower and upper Minkowski dimensions $\underline{d}_M$ and $\overline{d}_M$, respectively \cite{Robinson:dimensions,Fraser:dimensions}. 
The upper bound is attained by certain sets, but if $\widetilde{\cX}$ is a finite union of smooth submanifolds, then so is $\sqrt{\!\widetilde{\cX}}$, and their Minkowski dimensions coincide with the maximum manifold dimension of the pieces. 

The following lemma says that the metric space $\cP(\cY)$ with the purified distance is equivalent to $S_+(\cY,1)$ with the $L^2$-norm metric: 
\begin{lemma}
\label{lemma:purified-HS}
For two probability distributions $P,Q\in\cP(\cL)$ and their purified distance $p(P,Q) = \sqrt{1-F(P,Q)^2}$, it holds
\[
  p(P,Q) \leq \left| \sqrt{P}-\sqrt{Q} \right|_2 \leq \sqrt{2}p(P,Q).
\]
\end{lemma}
\begin{proof}
By direct calculation, we get
\[\begin{split}
   \left| \sqrt{P}-\sqrt{Q} \right|_2^2
    &= \sum_{\ell\in\cL} \left(\sqrt{P(\ell)}-\sqrt{Q(\ell)}\right)^2 \\
    &= \sum_{\ell\in\cL} \left(P(\ell)+Q(\ell) -2\sqrt{P(\ell)}\sqrt{Q(\ell)}\right) \\
    &= 2(1-F(P,Q)).
\end{split}\]
Now, since $1\leq 1+F(P,Q) \leq 2$, we obtain the lower bound $1-F(P,Q)^2$ and the upper bound $2\bigl(1-F(P,Q)^2\bigr)$ for the latter expression. 
\end{proof}

With these concepts, let us return to Equation \eqref{eq:log-fidelity-bound}, where in the third line we recognize the purified distance. Hence, with Lemma \ref{lemma:purified-HS}, 
\begin{equation}\begin{split}
  \label{eq:trace-distance-lower-spherization}
  -\ln\!\left(\! 1-\frac12\|W_{x^n}-W_{{x'}^n}\|_1 \!\right) 
   \!\geq& \frac12 \sum_{i=1}^n p\left(W_{x_i},W_{x_i'}\right)^2 \\
   \geq& \frac14 \sum_{i=1}^n \left| \sqrt{W_{x_i}}-\sqrt{W_{x_i'}}\right|_2^2 \\
   =&    \frac14 \left| \bigoplus_{i=1}^n \sqrt{W_{x_i}} - \bigoplus_{i=1}^n \sqrt{W_{x_i'}} \right|_2^2\!\!\!.
\end{split}\end{equation}
In other words, in the finite and countably infinite case we can lower-bound how (exponentially) close to $1$ the total variation distance between two product distributions is in terms of the direct sum of the spherisations of the tensor factors. 

On the other hand, we have observed that in an $(n,N,\lambda_1,\lambda_2)$\,-\,DI code, $\frac12\|W_{u_j}-W_{u_k}\|_1 \geq 1-\lambda_1-\lambda_2 =: \delta >0$ for all $j\neq k$. Hence, with $u_j=x^n$, $u_k={x'}^n$, we can upper-bound the total variation distance as follows: start by using the Fuchs-van de Graaf relation \eqref{eq:Classical_FvdG},
\[\begin{split}
\delta^2 
\leq \left( \frac12\|W_{x^n}-W_{{x'}^n}\|_1 \right)^2 
&\leq 1-F\left(W_{x^n},W_{{x'}^n}\right)^2\\
&=1 - \prod_{i=1}^n F(W_{x_i},W_{x_i'})^2,
\end{split}\]
and the identity $1-\prod_{i=1}^n a_i = \sum_{i=1}^n (1-a_i)\prod_{j=i+1}^n a_j$ with the convention that an empty product is $1$ and using it with $a_i = F(W_{x_i},W_{x_i'})^2 \in [0;1]$ (this identity is proved easily by induction or by expanding the right hand side terms and telescoping):
\begin{equation}\begin{split}
  \label{eq:trace-distance-upper-spherization}
  \delta^2 
   &\leq    1 - \prod_{i=1}^n F(W_{x_i},W_{x_i'})^2 \\
   &\leq \sum_{i=1}^n \left( 1-F(W_{x_i},W_{x_i'})^2 \right) \\
   &=    \sum_{i=1}^n p(W_{x_i},W_{x_i'})^2 \\
   &\leq \sum_{i=1}^n \left| \sqrt{W_{x_i}} - \sqrt{W_{x_i'}} \right|_2^2 \\
   &=    \left| \bigoplus_{i=1}^n \sqrt{W_{x_i}} - \bigoplus_{i=1}^n \sqrt{W_{x_i'}} \right|_2^2,
\end{split}\end{equation}
where the last inequality follows from Lemma \ref{lemma:purified-HS}. 

The important take-away from Equations \eqref{eq:trace-distance-lower-spherization} and \eqref{eq:trace-distance-upper-spherization} is that both the lower and upper bounds on the total variation distance are expressed in terms of the square of the same metric distance, which conveniently is the Euclidean distance on $\bigoplus_{i=1}^n \sqrt{W_{x_i}} \in S_+(\cY,1)^n \subset \RR^{n\cY}$. 

We remark here that also 
\begin{equation}
  \label{eq:d_p-definition}
  d_p\left( \bigoplus_{i=1}^n W_{x_i}, \bigoplus_{i=1}^n W_{x_i'} \right)
    := \sqrt{\sum_{i=1}^n p(W_{x_i},W_{x_i'})^2}
\end{equation}
defines a metric, this time on $\bigoplus_{i=1}^n W_{x_i} \in \cP(\cY)^n \subset \RR^{n\cY}$, which gives slightly tighter (but ultimately equivalent) bounds on the total variation distance. 
When $\cY$ is uncountable and the $W_x$ are all absolutely continuous with respect to an underlying measure $\lambda$, this metric still makes good sense. Namely, in such a case, the fidelity is still well-defined and given by 
\[
  F(W_x,W_{x'}) = \int_{\cY} \lambda({\rm d}y) \sqrt{\frac{{\rm d}W_x}{{\rm d}\lambda}(y)} \sqrt{\frac{{\rm d}W_{x'}}{{\rm d}\lambda}(y)}.
\]
By the way, also the above correspondence with a spherisation as well as Lemma \ref{lemma:purified-HS} can then be saved, by defining the functions $\sqrt{\frac{{\rm d}W_x}{{\rm d}\lambda}(y)}$, which are pointwise non-negative unit vectors in the Hilbert space $L^2(\cY,\lambda)$ of square-integrable functions on $\cY$ with respect to the measure $\lambda$.

\section[{Slightly superlinear DI capacity}]{Slightly superlinear DI capacity\protect\\ for channels with general input}
\label{sec:results}
In the present section we prove the main results of the paper. In the first two subsections we analyse two particular scenarios which can be solved in an easy manner using volume arguments on the output probability set. 
The first one studies the Bernoulli channel $B$ which has some additional interest as it can be simulated by any continuous channel $W$, meaning that the lower bound performance that we prove for $B$ can automatically be achieved by all other such channels $W$. The second studies the Poisson channel proving by very simple arguments that its superexponential capacity is upper bounded by $\dot{D}_\text{DI}(P)\leq\frac12$, improving the previously known bound of $\frac32$. They both serve as a prelude to the general case, which could seem unnecessarily abstract otherwise.

We move to the general case in Subsection \ref{subsec:main} using a more abstract construction with covering and packing arguments on a modified output probability space. We find bounds for the pessimistic and optimistic DI capacity in terms of the lower and upper Minkowski dimensions of the modified probability set respectively. Subsection \ref{subsec:d=0} analyses the surprisingly rich case where the Minkowski dimension is zero which allows us to further generalise the previous results into the very strong Theorem \ref{thm:Abstract_Generalization}. We also describe a scenario for superactivation in communication over memoryless and identically distributed classical channels in Subsection \ref{subsec:superactivation}. Finally, in Subsection \ref{subsec:example}, we study DI over a particular channel $A$ with continuous input and output which attains the upper bound for the capacity.

\subsection[{Bernoulli channel}]{Deterministic identification over the Bernoulli channel}
\label{subsec:bernoulli}
\begin{definition}
The \emph{Bernoulli channel} $B:[0;1] \rightarrow \{0,1\}$ on input $x\in[0;1]$ (a real number) outputs a binary variable according to Bernoulli distribution $B_x$ with parameter $x$:
\begin{equation}\label{eq:BernoulliChannel}
     B(y|x) = B_x(y) = xy + (1-x)(1-y)
      = \begin{cases}
          x   & \text{for } y=1, \\
          1-x & \text{for } y=0.
        \end{cases}
\end{equation}
\end{definition}
On block length $n$ we have a continuous set of inputs $x^n = x_1\dots x_n \in[0;1]^n$ which are points in the unit hypercube of dimension $n$. On the other hand, we have a finite set of $2^n$ possible outputs $y^n\in\{0,1\}^n$. 


\begin{theorem}
\label{thm:Bernoulli}
The superexponential deterministic identification capacity of the Bernoulli channel $B$ is bounded, for every $\lambda_1,\lambda_2>0$ with $\lambda_1+\lambda_2<1$, as
\[
  \frac14\leq\dot{C}_{\text{DI}}(B) 
    \leq \limsup_{n\rightarrow\infty} \frac{1}{n\log n} \log N_{\text{DI}}(n,\lambda_1,\lambda_2)\leq \frac12.
\]
\end{theorem}
\begin{proof}
We start with the lower bound (achievability), remembering Theorem \ref{thm:ID-from-packing} where we have seen that finding a $\sqrt{6\delta}\sqrt[4]{n}$-packing of points $\bigoplus_{i=1}^n B_{x_i}$ in the output probability set $\mathcal{P}(\mathcal{Y})^n$ given some $\delta\geq0$, with respect to the metric induced by the mixed $(1,2)$-norm is enough to prove the existence of a good DI code. We can simplify this condition thanks to the structure of the Bernoulli channel which has total variation distance between output distributions $B_x$ and $B_{x'}$ given by
\[\begin{split}
  \frac12\| B_{x}-B_{x'} \|_1 
    &= \sum_{y=0}^1 \frac12|B_x(y) - B_{x'}(y)|\\
    &= \frac12|1-x-1+x'|+\frac12|x-x'| = |x-x'|.
\end{split}\]
Thus, the mixed $(1,2)$-norm between Bernoulli distributions $B_{x^n}$ and $B_{{x'}^n}$ is equal to the euclidean norm between the corresponding input probability distributions $x^n$ and ${x'}^n$:
\begin{equation*}
\begin{split}
  \frac12 \left\| \bigoplus_{i=1}^n B_{x_i} - \bigoplus_{i=1}^n B_{x_i'} \right\|_{1,2} 
  &= \sqrt{\sum_{i=1}^n \left(\frac12\| B_{x_i}-B_{x_i'} \|_1\right)^2}\\
  &= \sqrt{\sum_{i=1}^n \lvert{x_i-x_i'}\rvert^2}\\ 
  &= \left\|x^n-{x'}^n\right\|_2.  
\end{split}
\end{equation*}
Now we just need an Euclidean $\sqrt{6\delta}\sqrt[4]{n}$-packing on the input sequence set $\cX^n=[0;1]^n$ to construct a good DI code. Inspired by \cite{DI-poisson}, we can count the elements of such a packing using volume arguments similar to the Minkowski-Hlawka theorem \cite{CS:Packings_lattices, Cohn:Order_Packing}. The packing can be studied as an arrangement $\{ S_{u_i}(n,r) \}$ of $N$ (large) disjoint open $n$-balls of radius $r=\sqrt{3\delta/2}\sqrt[4]{n}$ in the unit $n$-hypercube.

Specifically, consider a maximal packing arrangement $\mathcal{V}$ in $[0;1]^n$ with $N$ balls of radius $r$. Thus, no more balls with centers in $[0;1]^n$ can be added without overlap, meaning that no point in the hypercube is at a distance greater than $2r$ from all the sphere centers (or the packing would not be saturated, see Figure \ref{fig:extended_packing}). Therefore, by doubling the radius of the balls (which multiplies their total volume by $2^n$), we cover all the points of the hypercube. This means that the density $\Delta_n(\mathcal{V})$ of the original packing is at least $2^{-n}$:
\[
  \Delta_n(\mathcal{V}) 
  = \frac{\text{Vol}\left[C_0(n,1)\cap\bigcup_{i=1}^N S_{u_i}(n,r)\right]}{\text{Vol}\left[C_0(n,1)\right]}\geq2^{-n},
\]
where $\text{Vol}\left[C_0(n,s)\right]=s^n$ is the volume of the $n$-dimensional hypercube of side $s$. In our case ($s=1$), $\text{Vol}\left[C_0(n,1)\right]=1$. 
Also, each $n$-ball of radius $r$ and centered at ${u_i}\in\mathcal{X}^n$ is an $n$-dimensional hypersphere and therefore its volume is given by
\begin{equation}\label{eq:n-SphereVolume}
\text{Vol}\left[S_{u_i}(n,r)\right]=\frac{\pi^{n/2}r^n}{\Gamma\left(\frac{n}{2}+1\right)}.
\end{equation}
The number of spheres in the packing can thus be lower-bounded
\[\begin{split}
N&=\frac{\text{Vol}\left[\bigcup_{i=1}^N S_{u_i}(n,r)\right]}{\text{Vol}\left[S_{u_1}(n,r)\right]}\\
&\geq\frac{\text{Vol}\left[C_0(n,1)\cap\bigcup_{i=1}^N S_{u_i}(n,r)\right]}{\text{Vol}\left[S_{u_1}(n,r)\right]}\\
&=\Delta_n(\mathcal{V})\frac{\text{Vol}\left[C_0(n,1)\right]}{\text{Vol}\left[S_{u_1}(n,r)\right]}.
\end{split}\]
Inserting the known values of the density and the volumes, and taking the logarithm on both sides of the inequality, we get
\begin{equation}\begin{split}
  \label{eq:quarter-lower-bound}
  \log N 
   &\geq \log2^{-n}\frac{\Gamma\left(\frac{n}{2}+1\right)}{\left(2r\sqrt{\pi}\right)^n} \\
   &= \log \Gamma\left(\frac{n}{2}+1\right) -n\left[\log\left(2r\sqrt{\pi}\right)-1\right] \\
   &= \frac{n}{2}\log n - n\log\left(r\right) - O(n)\\
   &= \frac{n}{2}\log n - n\log\left(n^\frac14\right) - O(n)\\
   &=\frac{n}{4}\log n-O(n).
\end{split}\end{equation}
We find that the number of spheres in the packing grows superexponentially in $n$, the exponent scaling as $\frac14 n\log n$ to leading order. This sphere packing on the input guarantees that the packing conditions on the output needed to apply Theorem \ref{thm:ID-from-packing} are fulfilled, and therefore we can construct a good DI code with this superexponential size. 


For the converse part we start with an arbitrary $(n,N,\lambda_1,\lambda_2)$-DI code $\cC:=\{(u_j,\cE_j) : j\in[N]\}$, which by construction has to satisfy $\frac12\|B_{u_j}-B_{u_k}\|_1 \geq 1-\lambda_1-\lambda_2 =: 2\sqrt{2}r$ for any pair $j\neq k$. Let us express the code words as the sequences $u_j:=x^n=x_1\dots x_n$ and $u_k:=x'^n=x'_1\dots x'_n$. From this, invoking the bound \eqref{eq:trace-distance-upper-spherization} at the end of Section \ref{sec:sqrt-metric}, we get 
\begin{equation}
\label{eq:Bern_norm_packing}
  8r^2 \leq \sum_{i=1}^n \left| \sqrt{B_{x_i}} - \sqrt{B_{x'_i}} \right|_2^2,
\end{equation}
Now, for the Bernoulli channel the different $\sqrt{B_x}=(\sqrt{x},\sqrt{1-x})$ are unit vectors in the positive quadrant of the two-dimensional Euclidean plane. Let us further define the projection $\widetilde{B}_x=(\tilde{x},\sqrt{2}-\tilde{x})$ of each $\sqrt{B_x}$ from the origin onto the straight line connecting $(0,\sqrt{2})$ and $(\sqrt{2},0)$, see Figure \ref{fig:Bern_reparametrization}. Clearly $|\sqrt{B_x}-\sqrt{B_{x'}}|_2\leq|\widetilde{B}_x-\widetilde{B}_{x'}|_2=\sqrt{2}|\tilde{x}-\tilde{x'}|$. Notice also that $\tilde{x}$ is nothing but yet another reparametrisation of the input letter $x$. Indeed, it is clear that
\[
  \frac{\sqrt{1-x^2}}{x}
    = \frac{\sqrt{2}-\tilde{x}}{\tilde{x}},
    \text{ i.e., }
    \tilde{x}=\frac{\sqrt{2}x}{\sqrt{1-x^2}+x}.
\]

\begin{figure*}[ht]
\centering
\begin{minipage}{0.488\textwidth}
  \centering
    \includegraphics[width=0.768\linewidth]{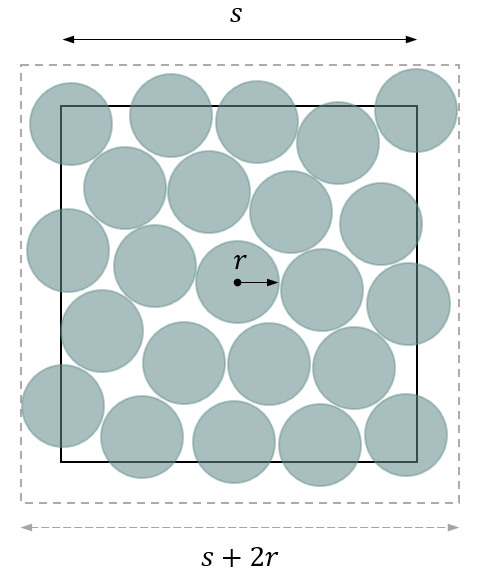}
    \caption{Packing in a cube of side $s$ ($s=1$ in the direct part), and an extension containing all balls. In the converse, $s=\sqrt{2}$.}
    \label{fig:extended_packing}
\end{minipage}%
\hspace{0.0162\textwidth}
\begin{minipage}{0.488\textwidth}
 \centering
 \vspace{0.53cm}
    \includegraphics[width=0.84\linewidth]{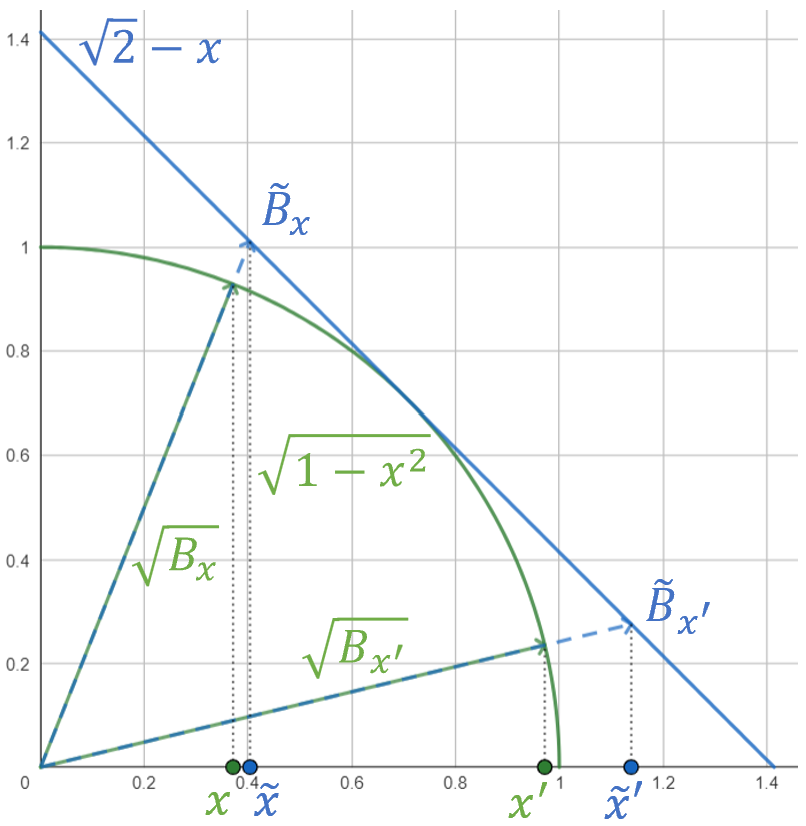}
    \caption{Projection of two unit vectors $\sqrt{B_x}$ and $\sqrt{B_{x'}}$ onto the tangent line $\sqrt{2}-x$, 
    and the extended vectors $\widetilde{B}_x$ and $\widetilde{B}_{x'}$.
    }
    \label{fig:Bern_reparametrization}
\end{minipage}
\end{figure*}
This allows us to lower-bound the Euclidean distance separating any two different reparametrised code words $\tilde{x}^n\neq\tilde{x}'^n\in(0;\sqrt{2})^n$ of the identification code. Continuing from Equation \eqref{eq:Bern_norm_packing}, we obtain: 
\begin{equation*}
\begin{split}
    2\sqrt{2}r
    \leq \sqrt{\sum_{i=1}^n\left|\sqrt{B_{x_i}}-\sqrt{B_{x'_i}}\right|_2^2}
    &\leq \sqrt{\sum_{i=1}^n2\left|\tilde{x_i}-\tilde{x'_i}\right|^2}\\
    &= \sqrt{2}\left| \tilde{x}^n-\tilde{x'}^n \right|_2.
\end{split}
\end{equation*}
We can immediately define an arrangement of disjoint balls $S_{u_j}(n,r)$ of radius $r$ around each code word $u_j\in\cC$ and use volume arguments to upper bound the number $N$ of balls (and therefore, the number of code words) that can be packed into the hypercube $C(n,\sqrt{2})$ (the reparametrised input space). To do that, we define the extended hypercube $C(n,\sqrt{2}+2r)$ such that it contains all possible balls centered inside the input space (see Figure \ref{fig:extended_packing} below). Then, $\text{Vol}[C(n,\sqrt{2}+2r)]\geq\bigcup_{u_i=1}^{N} \text{Vol}[S_{u_i}(n,r)]$. As all the balls have the same radius, and therefore the same volume, we can deduce
\[N\leq\frac{\text{Vol}[C(n,\sqrt{2}+2r)]}{\text{Vol}[S_{u_i}(n,r)]}=\left(\frac{\sqrt{2}+2r}{r\sqrt{\pi}}\right)^n\Gamma\left(\frac{n}{2}+1\right).\]
We conclude by taking the logarithm,
\(
\log N\leq\frac{n}{2}\log n+n\log\frac{\sqrt{2}+2r}{r\sqrt{\pi}}=\frac{n}{2}\log n+O(n).
\)
\end{proof}

The Bernoulli channel is a very significant example, not only in itself but also when restricting the input to a subset $\cX \subset [0;1]$, i.e. $B\vert_{\cX} : \cX \rightarrow \{0,1\}$ (see Subsection \ref{subsec:d=0} below). In particular, the case $\cX=[a;b]\subset[0;1]$ with $0\leq a < b \leq 1$ is important, as we shall discuss now. 

\begin{corollary}
\label{cor:Bernoulli-restricted}
Just as in Theorem \ref{thm:Bernoulli}, 
\[
  \frac14\leq\dot{C}_{\text{DI}}(B\vert_{[a;b]}) 
    \leq \limsup_{n\rightarrow\infty} \frac{1}{n\log n} \log N_{\text{DI}}(n,\lambda_1,\lambda_2)\leq \frac12.
\]
\end{corollary}
\begin{proof}
The upper (converse) bound is clear, since every DI code for the restricted Bernoulli channel $B\vert_{[a;b]}$ is one for the unrestricted Bernoulli channel. 
For the direct part, notice that we can simply repeat the packing argument from the proof of Theorem \ref{thm:Bernoulli}, only with the smaller hypercube $[a;b]^n \subset [0;1]^n$. This will affect the $O(n)$ term of the code size bound \eqref{eq:quarter-lower-bound}, but not the leading order term $\frac{1}{4}n\log n$.
\end{proof}

To understand how relevant this is, observe that any channel with a truly continuous input, in the sense that $\widetilde{\cX} = W(\cX) \subset \cP(\cY)$ contains a continuous curve, can simulate the Bernoulli channel restricted to a subinterval $[a;b]\subset[0;1]$, via suitable classical pre- and post-processing. Namely, consider such a curve $\mathfrak{C} = \{W_{x_t} : 0 \leq t \leq 1\} \subset \widetilde{\cX}$ (the image of the unit interval under the continuous mapping $t \mapsto W_{x_t}$) with its end points $P=W_{x_0}$ and $Q=W_{x_1}$. As these points are distinct, there exists a binary test $\cY = \cY_0 \,\stackrel{.}{\cup}\, \cY_1$ (with $\stackrel{.}{\cup}$ the disjoint union) distinguishing them statistically, w.l.o.g. $a := P(\cY_1) < Q(\cY_1) =: b$. This test can be written as a deterministic channel $T:\cY \rightarrow \{0,1\}$ acting as $T(y) = z$ iff $y\in\cY_z$. The images of $P$ and $Q$ under this channel are $T(P)=B_a$, $T(Q)=B_b$. By continuity and the intermediate-value theorem, for every $s\in[a;b]$ there is a $t\in[0;1]$ such that $W_{x_t}(\cY_1) = s$, i.e. $T(W_{x_t}) = B_s$. In other words, $T(W(\cX))$ contains the image of $B\vert_{[a;b]}$, and hence $\dot{C}_{\text{DI}}(W) \geq \dot{C}_{\text{DI}}(T\circ W) \geq \dot{C}_{\text{DI}}(B\vert_{[a;b]}) \geq \frac14$. Thus we have proved the following: 

\begin{corollary}
\label{thm:continuous}
Consider any channel $W:\cX\rightarrow\cY$ such that $\widetilde{\cX} = W(\cX) \subset \cP(\cY)$ contains a continuous curve. Then, its superexponential deterministic identification capacity is lower-bounded as 
\(
  \dot{C}_{\text{DI}}(W) \geq \frac14. 
\)
\qed
\end{corollary}

It is noteworthy that this corollary makes no assumptions on the alphabet $\cY$. On the other hand, for discrete $\cY$ we can drop the continuity assumption and will get a much deeper understanding and improvement of the lower bound, in Subsection \ref{subsec:main} below. 
However, for Gaussian channels (AWGN or general fading) as well as Poisson channels (see the next Subsection), both with peak as well as average power constraints, the same reasoning as in Corollary \ref{cor:Bernoulli-restricted} applies. To realise the above argument, it is enough to fix a line segment in the input parameter space that satisfies the power constraints in its entirety. By reduction to the input-restricted Bernoulli channel we thus reproduce the prior achievability results $\dot{C}_{\text{DI}}(G) \geq \frac14$ \cite{SPBD:DI_power,DI-fading} and $\dot{C}_{\text{DI}}(P) \geq \frac14$ \cite{DI-poisson,DI-poisson_mc}.

\subsection[{Poisson channel}]{Deterministic identification over the Poisson channel}
\label{subsec:poisson}
\begin{definition}
\label{def:Poisson-channel}
The \emph{discrete time Poisson channel} $P:\RR_{\geq 0} \rightarrow \NN_0$ on a positive real input $x \in \cX = \RR_{\geq0}$ outputs a non-negative integer $y \in \cY = \NN_0$ with Poisson distributed probability
\begin{equation}
  P_x(y) = P(y|x) = \frac{e^{-x}x^y}{y!}.
\end{equation}
\end{definition}
The channel is memoryless for $n\in\NN$ uses of the channel, with continuous inputs $x^n=x_1\dots x_n\in\RR_{\geq 0}^n$ and outputs $y^n=y_1\dots y_n\in\NN_0^n$. 
The Poisson channel arises as a key channel model for practical 6G networks in the context of molecular \cite{DI-poisson_mc,GMN:molecularComs} and optical communications \cite{Cao:PhDThesis}. Deterministic identification over the discrete time Poisson channel $P$ was studied in \cite{DI-poisson}, and the following superexponential DI capacity bounds under peak power ($x^n$ such that for all $i$, $x_i\leq E_{\max}$) and average power ($x^n$ such that $\sum_i x_i \leq nE$) constraints were found:
\[
  \frac14 \leq \dot{C}_\text{DI}(P) \leq \frac32.
\]
Inspired by the tools in the previous Subsection \ref{subsec:bernoulli}, we present a simplified converse proof that brings us to an improvement on the upper bound. 

The lower bound from the first DI analysis of the Poisson channel \cite{DI-poisson} works when imposing both average and peak power constraints. On the other hand, our new improved converse proof can be shown without the need of a peak power constraint. So we relax the hypothesis by imposing only an average power constraint: given a positive energy $E > 0$, we demand that $\sum_{i=1}^nx_i\leq nE$ for every code word $u_i=x^n$.

\begin{theorem}
\label{thm:Poisson-improved}
The superexponential deterministic identification capacity over the Poisson channel $P$ is bounded, for every $\lambda_1,\lambda_2>0$ with $\lambda_1+\lambda_2<1$, as
\[
    \frac{1}{4}\leq \dot{C}_\text{DI}(P)\leq\limsup_{n\rightarrow\infty}\frac{1}{n\log n}\log N_\text{DI}(n,\lambda_1,\lambda_2)\leq\frac12.
\]
\end{theorem}
\begin{proof}
Only the converse part is needed. We follow the first steps from the previous section, starting with an $(n,N,\lambda_1,\lambda_2)$-DI code for the Poisson channel. By definition, the output distributions of distinct code words $x^n$ and $x'^n$ must have $\frac12\|P_{x^n}-P_{x'^n}\|_1\geq1-\lambda_1-\lambda_2=:\delta$. We use again the distance relations in Equation~\eqref{eq:Classical_FvdG}, to rewrite this in terms of fidelity between code word letters:
\[\begin{split}
    (2r)^2:=-\ln(1-\delta^2)&\leq-\ln[F(P_{x^n},P_{x'^n})^2]\\
    &=\sum_{i=1}^n-\ln[F(P_{x_i},P_{x'_i})^2].
 \end{split}\]
Now notice that the fidelity between two Poisson distributions is given by
\begin{equation*}
\begin{split}
F(P_{x_i},P_{x'_i})^2
  &=\left[\sum_{y\in\NN_0}\sqrt{\frac{e^{-x_i} x_{i}^{y}}{y!}\frac{e^{-x'_{i}} x'_{i}{}^{y}}{y!}}\right]^2\\
  &=\left[e^{-\frac{x_i+x'_i}{2}}\sum_{y\in\NN_0}\frac{\sqrt{x_ix'_i}^y}{y!}\right]^2\\
  &=\exp\left[-x_i-x'_i+2\sqrt{x_ix'_i}\right]\\
  &=\exp\left[-\left(\sqrt{x_i}-\sqrt{x'_i}\right)^2\right],
\end{split}
\end{equation*}
which means that 
\(
  (2r)^2 
  \leq \sum_{i=1}^n \left(\sqrt{x_i}-\sqrt{x'_i}\right)^2.
\)
Let us now reparameterise each letter defining $s_i:=\sqrt{x_i}$. Then we have a minimum Euclidean distance between the code words $s^n$ and $s'^n$ (with $s^n=s_i\dots s_n$):
\begin{equation}
  \label{eq:PoissonDistance}
  2r\leq\sqrt{\sum_{i=1}^n|s_i-s'_i|^2} 
      = \|{s^n-s'^n}\|_2.
\end{equation}
\begin{figure}[ht]
    \centering
    \includegraphics[scale=0.26]{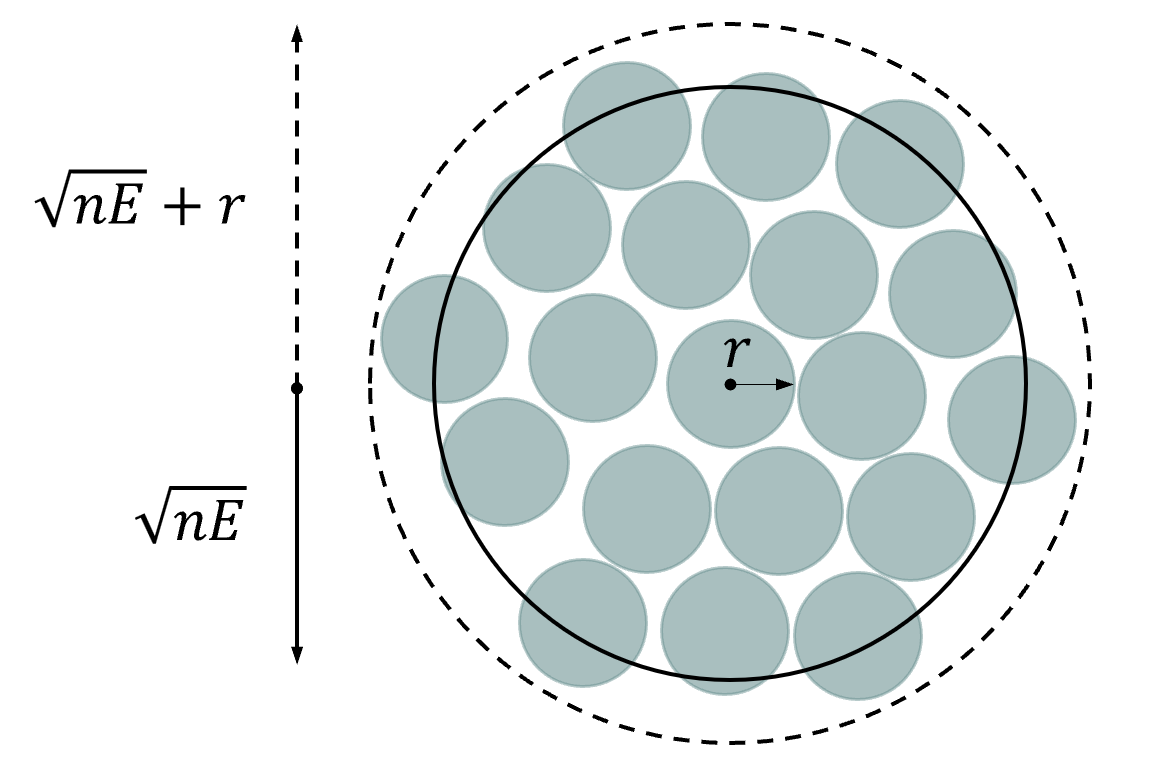}
    \caption{An $r$-ball packing with centers in the big sphere $S(n,\sqrt{nE})$ imposed by the power constraint, and the extended sphere $S(n,\sqrt{nE}+r)$ that contains all.}
    \label{fig:poisson}
\end{figure}

Also, the power constraint becomes $\sum_{i=1}^ns_i^2\leq nE$, which imposes that all code words lie inside an $n$-hypersphere $S_0(n,\sqrt{nE})$ of radius $\sqrt{nE}$ around the origin. 
We can create a set of non-intersecting balls $S_{u_i}(n,r)$ of radius $r$ centered at each code word (a ball packing inside the hypersphere that contains all code words under the power constraint, see Figure \ref{fig:poisson}). We use volume arguments to calculate the number of elements $N$ in such a packing. Clearly 
\(
\text{Vol}[S(n,\sqrt{nE}+r)]\geq \text{Vol}\left[\bigcup_{u_i=1}^{N}S_{u_i}(n,r)\right],
\)
and we know the volume of the hyperspheres [Equation~\eqref{eq:n-SphereVolume}], so
\[
N\leq\frac{\text{Vol}[S(n,\sqrt{nE}+r)]}{\text{Vol}[S_{u_i}(n,r)]}=\frac{(\sqrt{nE}+r)^n}{r^n}\leq\frac{(2\sqrt{nE})^n}{r^n}.
\]
We complete the proof by taking the logarithm and bounding
\[
\log N \leq\frac{n}{2}\log n+\frac{n}{2}\log E+n-n\log r=\frac{n}{2}\log n + O(n).
\]
That is, we arrive at the upper bound 
\(
\dot{C}_\text{DI}(P)\leq\frac12,
\)
tightening the previous converse \cite{DI-poisson,DI-poisson_mc}.
\end{proof}

\subsection[{General case: pessimistic/optimistic capacities}]{Pessimistic and optimistic DI capacities in the general case}
\label{subsec:main}
The previous particular cases can be solved using volume arguments on the input set thanks to the particular structure of the channels, which allowed us to move from some distance between output probability distributions (which define the DI code) to an Euclidean distance between the corresponding input sequences (or a reparametrisation of those), which is required to define volumes on the input. 
However, this transformation is not available for most channels, meaning that, in general, we do not have access to volume arguments on the level of the code words to bound the number of elements in a code. We show that we can use instead Euclidean covering and packing arguments on the single-letter level, which lead to general bounds for the pessimistic and the optimistic capacities in terms of the lower and upper Minkowski dimensions, respectively. These bounds generalise all superexponential deterministic identification results to date. 

\begin{theorem}
\label{thm:pessimistic_final}
The superexponential deterministic identification capacity $\dot{C}_{\text{DI}}(W)$ for the channel $W:\cX\rightarrow\cY$ is bounded in terms of the lower Minkowski dimension of the set $\sqrt{\!\widetilde{\cX}}$ as follows (for any $\lambda_1,\lambda_2>0$ with $\lambda_1+\lambda_2<1$):
\[\begin{split}
\frac14\underline{d}_M\!\left(\!\sqrt{\!\widetilde{\cX}}\right)
    \leq \dot{C}_{\text{DI}}(W)
    &\leq \liminf_{n\rightarrow\infty}\frac{1}{n\log n}\log N_{\text{DI}}(n,\lambda_1,\lambda_2)\\
    &\leq \frac12\underline{d}_M\!\left(\!\sqrt{\!\widetilde{\cX}}\right).
\end{split}\]
\end{theorem}
\begin{proof}
We begin analyzing the strong converse with an arbitrary $(n,N,\lambda_1,\lambda_2)$-DI code $\cC=\{(u_j,\cE_j) : j\in[N]\}$, which by construction has $\frac12\|{W_{u_j}-W_{u_k}}\|_1 \geq 1-\lambda_1-\lambda_2 =: 3r$ for any two different code words $u_j=x_1\dots x_n$ and $u_k = x'_1\dots x'_n$ in $\cX^n$. We invoke now the inequality \eqref{eq:trace-distance-upper-spherization} to deduce
\[
  (3r)^2 \leq \frac12\|{W_{u_j}-W_{u_k}}\|_1^2
   \leq \sum_{i=1}^n \left|\sqrt{W_{x_i}}-\sqrt{W_{x'_i}}\right|_2^2.
\]
In other words, we find that the $\bigoplus_{i=1}^n \sqrt{W_{x_i}}$ form a packing with constant separation $3r$. For a general set $\widetilde{\cX}$, we do not have access to the volume arguments of the Bernoulli case, so we choose an $\frac{r}{\sqrt{n}}$-covering (with respect to the Euclidean distance) of $\sqrt{\!\widetilde{\cX}}$, i.e.~a subset $\cX_{r\!/\!\sqrt{n}}\subset\cX$ such that for every $x\in\cX$ there is a $\xi\in\cX_{r\!/\!\sqrt{n}}$ with $\left| \sqrt{W_{x}}-\sqrt{W_{\xi}} \right|_2 \leq \frac{r}{\sqrt{n}}$. Thus, for the code word $u=x^n=x_1\ldots x_n$ we find another $u'=\xi^n=\xi_1\ldots\xi_n\in\cX_{r\!/\!\sqrt{n}}^n$ such that 
\begin{equation}
  \label{eq:last:common_converse_step}
  \sqrt{ \sum_{i=1}^n \left| \sqrt{W_{x_i}}-\sqrt{W_{\xi_i}} \right|_2^2 } \leq r.
\end{equation}
We can now define a modified code that uses the elements in the covering as the new code words $\{(u'_j,\cE_j) : j\in[N]\}$ which will have first and second kind of errors given by $\lambda'_1\leq\lambda_1+r$ and $\lambda'_2\leq\lambda_2+r$. Therefore, 
\[
\lambda'_1+\lambda'_2\leq\lambda_1+\lambda_2+2r=1-r<1,
\]
implying that different messages $j\neq k$ must have different encodings $u'_j\neq u'_k\in\cX_{r\!/\!\sqrt{n}}^n$, so the mapping from $\cC \ni u_j \mapsto u_j' \in \cX_{r\!/\!\sqrt{n}}^n$ is one-to-one, and hence $N\leq|\cX_{r\!/\!\sqrt{n}}|^n$. The only thing remaining is to upper-bound the size of the covering in terms of the lower Minkowski dimension. 
To do that, let us choose a sequence of $\delta_k>0$ converging to $0$ in such a way that 
\[
  \lim_{k\rightarrow\infty} \frac{\log\Gamma_{\delta_k}\!\left(\!\sqrt{\!\widetilde{\cX}}\right)}{-\log\delta_k} 
  = \underline{d}_M\!\left(\!\sqrt{\!\widetilde{\cX}}\right) =: d.
\]
Now let $n_k=\lfloor r^2/\delta_k^2\rfloor$ be the corresponding block length, so that the $\delta_k$-covering is also a $\frac{r}{\sqrt{n_k}}$-covering $\cX_{r\!/\!\sqrt{n_k}}$. 
By definition of the covering dimension, for every $\epsilon>0$ there is a $K>0$ such that we can find coverings with 
\[
  |\cX_{r\!/\!\sqrt{n_k}}| 
  \leq \left(\frac{K}{\delta_k}\right)^{d+\epsilon}
  \leq \left(\frac{K\sqrt{2n_k}}{r}\right)^{d+\epsilon}.
\]
Subsequently, a code of block length $n_k$ can have at most $N\leq|\cX_{r\!/\!\sqrt{n_k}}|^{n_k}$ elements, and that being so, we obtain 
$\log N \leq \frac12(d+\epsilon)n_k\log n_k+O(n_k)$. Finally, as $\epsilon$ can be made arbitrarily small and $(n_k)$ is a subsequence of the natural numbers converging to infinity, we find
\[\begin{split}
  \dot{C}_{\text{DI}}(W) 
   &= \liminf_{n\rightarrow\infty}\frac{\log N_{\text{DI}}(n,\lambda_1,\lambda_2)}{n \log n}\\
   &\leq \liminf_{k\rightarrow\infty} \frac{\log N_{\text{DI}}(n_k,\lambda_1,\lambda_2)}{n_k \log n_k}
   \leq \frac12 d,
\end{split}\]
completing the converse proof.

For the direct part, at block length $n$, we start with an Euclidean packing of $\sqrt{\!\widetilde{\cX}}$, i.e.~a subset $\cX_{n^{-\alpha}}\subset\cX$ of the input space such that for each pair $x\neq x'\in\cX_{n^{-\alpha}}$,
\[
  \lvert{\sqrt{W_x}-\sqrt{W_{x'}}}\rvert_2\geq n^{-\alpha}.
\]
By definition of the lower Minkowski dimension, that we now redefine as $d=\underline{d}_M\!\left(\!\sqrt{\!\widetilde{\cX}}\right)$, for every $\epsilon > 0$ there is a positive constant $K>0$ such that we can find such a packing with $\lvert{\cX_{n^{-\alpha}}}\rvert \geq (Kn^{\alpha})^{d-\epsilon}$ elements. Then, for $x^n,{x'}^n\in\cX_{n^{-\alpha}}^n$, we have, following Equation \eqref{eq:trace-distance-lower-spherization},
\begin{equation}
\label{eq:Main_achievability_Eucl}
\begin{split}
  -\ln\!\left(\! 1-\frac12\|W_{x^n}-W_{{x'}^n}\|_1 \!\right) 
   &\geq \sum_{i=1}^n \!-\ln F\left(W_{x_i},W_{x_i'}\right) \\
   &= \frac12 \sum_{i=1}^n -\ln F\left(W_{x_i},W_{x_i'}\right)^2 \\
   &\geq \frac12 \sum_{i=1}^n \!\left[1-F\left(W_{x_i},W_{x_i'}\right)^2\right]\\
   &\geq \frac14 \sum_{i=1}^n \left|\sqrt{W_{x_i}} -\sqrt{W_{x_i'}}\right|_2^2.
\end{split}
\end{equation}
By construction, each term in the latter sum is either $0$ or at least $n^{-2\alpha}$. Let us now choose a code $\cC_t\subset\cX_{n^{-\alpha}}^n$ with a minimum Hamming distance between its elements:
\[
\forall x^n\neq {x'}^n\in\cC_t,\quad d_H(x^n,{x'}^n) \geq tn,
\]
where $t>0$ is an arbitrarily small constant. We immediately see that for code words in $\cC_t$, the last sum in Equation~\eqref{eq:Main_achievability_Eucl} will have non-zero elements in at least $tn$ positions, and for each one of those we can lower-bound the norm squared by $n^{-2\alpha}$, so
\begin{equation}
\label{eq:Main_achiev_exponent}
-\ln\left(1-\frac12\|W_{x^n}-W_{{x'}^n}\|_1\right)\geq\frac14 \left(n^{-2\alpha}\right)tn=\frac{t}{4}n^{1-2\alpha}.
\end{equation}
Now, we want to apply Theorem \ref{thm:ID-from-packing} which constructs an identification code from a packing on the output probability set. To satisfy the conditions of said theorem we require a scaling of at least $\sqrt{n}$ in the exponent of Equation~\eqref{eq:Main_achiev_exponent}. So, we need $1-2\alpha\geq\frac12$, i.e.~$\alpha\leq\frac14$ and then, by virtue of Theorem \ref{thm:ID-from-packing}, it is possible to find asymptotically good deterministic identification codes having $N = \left\lfloor \lvert{\cC_t}\rvert/(n\log\lvert{\cY}\rvert) \right\rfloor$ messages.

It is only left to bound the size of the code $\cC_t$. Indeed, by elementary combinatorics, the Hamming ball around any point in $\cX_{n^{-\alpha}}^n$ of radius $tn$ has 
$\leq \binom{n}{tn} |\cX_{n^{-\alpha}}|^{tn}$ elements. Hence, any maximal code of Hamming distance $tn$ has a number of code words at least the ratio between the total number of elements and the size of the Hamming ball (otherwise the code could be extended):
\[
  \lvert{\cC_t}\rvert \geq 
  \frac{\lvert{\cX_{n^{-\alpha}}}\rvert^n}{\binom{n}{tn}\lvert{\cX_{n^{-\alpha}}}\rvert^{tn}}
  \geq    2^{-n} \lvert{\cX_{n^{-\alpha}}}\rvert^{n(1-t)}.
\]
The Gilbert-Varshamov bound \cite{Gilbert:combinatorics,Var:combinatorics} guarantees almost the same (and asymptotically equivalent) performance, by a linear code over the prime field $\mathbb{F}_p$, after we trim down $\cX_{n^{-\alpha}}$ to the nearest smaller prime cardinality $p\geq \frac12|\cX_{n^{-\alpha}}|$ (Bertrand's postulate).
By direct substitution of the bound $|\cX_{n^{-\alpha}}| \geq \left({Kn^\alpha}\right)^{d-\epsilon}$ we obtain
\begin{equation}\begin{split}
\label{eq:general_code_size}
  N \geq \frac{\lvert{\cC_t}\rvert}{n\log|\cY|}
   &\geq 2^{-n-\log n-\log\log|\cY|} \left[\left(Kn^\alpha\right)^{d-\epsilon}\right]^{n(1-t)}\\
   &\geq 2^{\alpha(1-t)(d-\epsilon)n\log n -O(n)}.
\end{split}\end{equation}
Thus, also in the general case, we can achieve superexponential ($n\log n$) scaling in the block length, as long as the lower Minkowski dimension is positive. Above we have inferred that any $\alpha<\frac14$ is suitable to guarantee asymptotically vanishing errors, and $\epsilon,t>0$ are arbitrarily small, showing that any $n\log n$-rate below $\frac14 d$ is attainable. Putting everything together, we get 
\begin{equation}
\label{eq:main_achiev_capacity}
\begin{split}
    \dot{C}_\text{DI}(W)
    &= \inf_{\lambda_1,\lambda_2>0} \liminf_{n\rightarrow\infty}\frac{\log N_\text{DI}(n,\lambda_1,\lambda_2)}{n\log n}\\
    &\geq \liminf_{n\rightarrow\infty}\frac{N}{n\log n}
    \geq \frac14 \underline{d}_M\!\left(\!\sqrt{\!\widetilde{\cX}}\right),
\end{split}
\end{equation}
concluding the proof. 
\end{proof}

\begin{remark}
\label{Remark:direct_part}
Observe that if all we wanted was a packing of the $W_{u_j}$ with respect to the trace distance, we could obtain a lower bound of $\log N \geq \frac12 \underline{d}_M\!\left(\!\sqrt{\!\widetilde{\cX}}\right) n\log n -O(n)$. This is because in the direct part above, we could choose $\alpha$ arbitrarily close to $\frac12$, meaning that the converse in Theorem \ref{thm:pessimistic_final} above would be optimal.
The smaller prefactor of $\frac14$ (corresponding to $\alpha<\frac14$) for a DI code is due to the packing having to have distance $\Omega(\sqrt{n})$, so as to be able to apply Theorem \ref{thm:ID-from-packing}. Perhaps a more direct approach to the code construction could deal with that.
\end{remark}

This result generalises the Bernoulli and Poisson examples from Sections \ref{subsec:bernoulli} and \ref{subsec:poisson} respectively (which both have dimension $d=1$), albeit only for finite output alphabet, and gives us general bounds for the pessimistic identification capacity of a general channel in terms of the lower Minkowski dimension of its square root output probability distributions. We can similarly bound the optimistic identification capacity of a channel $W$ with discrete output in the superexponential scale $\dot{C}_{\text{DI}}^{\text{opt}}(W)$ (defined with the $\limsup$ instead of the $\liminf$ for $n\rightarrow\infty$) in terms of the upper Minkowski dimension of the square root map of the input image $\widetilde{\cX}$:

\begin{theorem}
\label{thm:optimistic_final}
The superexponential optimistic identification capacity $\dot{C}_{\text{DI}}^{\text{opt}}(W)$ of a channel $W:\cX\rightarrow\cY$ is bounded in terms of terms of the upper Minkowski dimension of $\sqrt{\!\widetilde{\cX}}$ as follows (for any $\lambda_1,\lambda_2>0$ with $\lambda_1+\lambda_2<1$):
\[\begin{split}
  \frac14\overline{d}_M\!\left(\!\sqrt{\!\widetilde{\cX}}\right)
    \leq \dot{C}_{\text{DI}}^{\text{opt}}(W) 
    &\leq \limsup_{n\rightarrow\infty}\frac{1}{n\log n}\log N_{\text{DI}}(n,\lambda_1,\lambda_2)\\
    &\leq \frac12\overline{d}_M\!\left(\!\sqrt{\!\widetilde{\cX}}\right).
\end{split}\]
\end{theorem}
\begin{proof}
For the converse part we retrace the same steps as in the proof of Theorem \ref{thm:pessimistic_final}, creating the same $\frac{r}{\sqrt{n}}$-covering of $\sqrt{\!\widetilde{\cX}}$ that yields the new code $\{(u'_j,\cE_j),j\in[N]\}$ with code words $u'_j\in\cX_{r\!/\!\sqrt{n}}^n$, meaning that $N\leq|\cX_{r\!/\!\sqrt{n}}|^n$. Now, we need to upper-bound $|\cX_{r\!/\!\sqrt{n}}|$ in terms of the upper (rather than the lower) Minkowski dimension, which actually simplifies the analysis. 
For the upper covering dimension (defined through the $\limsup$) we have for every $\epsilon>0$ a positive constant $K>0$ such that a covering with $\lvert{\cX_{r\!/\!\sqrt{n}}}\rvert\leq\left(\frac{K\sqrt{n}}{r}\right)^{d+\epsilon}$ exists, where $d:=\overline{d}_M(\sqrt{\!\widetilde{\cX}})$. Then $\log N\leq n\log|\cX_{r\!/\!\sqrt{n}}|\leq\frac12(d+\epsilon)n\log n+O(n)$. As $\epsilon$ can be made arbitrarily small, we find
\[
 \dot{C}_{\text{DI}}^\text{opt}(W) 
      \leq \limsup_{n\rightarrow\infty} \frac{\log N_{\text{DI}}(n,\lambda_1,\lambda_2) }{n\log n}
      \leq \frac12 \overline{d}_M\!\left(\!\sqrt{\!\widetilde{\cX}}\right).
\]

For the direct part, we want to find the upper Minkowski dimension as a lower bound. We use the same strategy as in the direct part proof in Theorem \ref{thm:pessimistic_final}, but now we need to lower bound through the $\limsup$. As in the converse proof of the pessimistic capacity (Theorem \ref{thm:pessimistic_final}), we again start with a suitable sequence of $\delta_k > 0$ converging to $0$, such that 
\[
  \lim_{k\rightarrow\infty} \frac{\log\Pi_{\delta_k}\!\left(\!\sqrt{\!\widetilde{\cX}}\right)}{-\log\delta_k} = \overline{d}_M\!\left(\!\sqrt{\!\widetilde{\cX}}\right) =: d.
\]
Now, choosing $n_k=\lceil\delta_k^{-1/\alpha}\rceil$ and following the direct part of Theorem \ref{thm:pessimistic_final}, we can create good identification codes of block length $n_k$ (whenever $\alpha < \frac14$ and $\epsilon>0$) of message length $\log N \geq \alpha (d-\epsilon) n_k\log n_k - O(n_k)$. 
Thus, 
\[\begin{split}
  \dot{C}_{\text{DI}}^{\text{opt}}(W)
    &=\limsup_{n\rightarrow\infty} 
      \frac{\log N_{\text{DI}}(n,\lambda_1,\lambda_2)}{n\log n}\\
    &\geq \limsup_{k\rightarrow\infty} \frac{\log N_{\text{DI}}(n,\lambda_1,\lambda_2)}{n_k\log n_k}
    \geq \frac14 \overline{d}_M\!\left(\!\sqrt{\!\widetilde{\cX}}\right),
\end{split}\]
concluding the proof. 
\end{proof}

\begin{corollary}
\label{cor:Bernoulli_new}
Given the channel $W:\cX\rightarrow\cY$, if $\widetilde{\cX}=W(\cX)\subset\cP(\cY)$ is a union of finitely many smooth submanifolds of maximum dimension $d$, then 
\[\hspace{1.7cm}
  \frac14 d\leq\dot{C}_{\text{DI}}(W) 
    \leq \dot{C}_{\text{DI}}^\text{opt}(W)
    \leq {\frac12}d. \hspace{1.7cm}\qed
\]
\end{corollary}

\subsection{Subtleties in dimension zero}
\label{subsec:d=0}
Theorems \ref{thm:pessimistic_final} and \ref{thm:optimistic_final} are valid even when the lower (for the pessimistic case) or the upper Minkowski dimensions (for the optimistic) of $\widetilde{\cX}$ are zero. In that case, the achievability bound is trivial, while the converse says that whatever the code size $N$, its $(n\log n)$-rate converges to zero. However, these statements clearly do not capture what is actually going on. Excluding the trivial case of a constant channel $W$, we always have \emph{some} distinct output distributions, $|\widetilde{\cX}| \geq 2$. By restricting the input to a suitable finite set, $W$ is transformed into a DMC as studied in \cite{SPBD:DI_power}, with two or more distinct rows, hence the code size scales at least exponentially, $N_{\text{DI}}(n,\lambda_1,\lambda_2) \geq 2^{cn}$, and in fact \cite[Corollary~5]{SPBD:DI_power} (cf.~\cite{AD:ID_ViaChannels,AC:DI}) shows that the linear-order DI-capacity $C_{\text{DI}}(W)$ and its strong converse for finite channel output $\widetilde{\cX}$ are given by
\[
  C_{\text{DI}}(W)=\lim_{n\rightarrow\infty} \frac1n \log N_{\text{DI}}(n,\lambda_1,\lambda_2) 
  = \log\,\bigl|\widetilde{\cX}\bigr|.
\]

However, for infinite $\widetilde{\cX}$, on which case we will focus from now on, we learn that for every $R>0$ there is an $n(R)$ such that for all $n\geq n(R)$, $N(n,\lambda_1,\lambda_2) \geq 2^{nR}$. On the other hand, there are infinite sets $\widetilde{\cX}\subset\cP(\cY)$ with upper (and therefore also lower) Minkowski dimension zero, and for those examples the considerations and theorems in the previous Section \ref{subsec:main} tell us
\begin{equation}
  \omega(n) 
  \leq \log N_{\text{DI}}(n,\lambda_1,\lambda_2) 
  \leq o(n\log n), 
\end{equation}
where we have used little-$o$ and little-$\omega$ notation to denote the asymptotic separation. 

Arguably the simplest examples, useful to understand what is happening here, are countable sets $\widetilde{\cX}$, without loss of generality closed. Since $\cP(\cY)$ is compact, such a set will have at least one accumulation point (an element $x\in\widetilde{\cX}$ that is in the closure of its relative complement $\widetilde{\cX}\setminus\{x\}$), and so again the simplest class consists of those with a unique accumulation point. Such a set is conveniently described by a convergent sequence $(x_t)_{t\in\NN}$ and its limit $x_* = \lim_{t\rightarrow\infty} x_t$: 
\[
  \widetilde{\cX} = \{ x_1,x_2,\ldots,x_t,\ldots \} \stackrel{.}{\cup} \{ x_* \}.
\]
For binary output, $\cP(\{0,1\}) \simeq [0;1]$, consider thus monotonically decreasing sequences of positive numbers $x_t\in(0;1]$ converging to $x_*=0$. This corresponds to restricting the Bernoulli channel $B$ to inputs from a subset of the unit interval. It turns out that the Minkowski dimension is related to the speed or slowness of this convergence, as we can illustrate by the following cases,
\begin{align*}
  L   &:= \{0\} \cup \left\{ (\log t)^{-1} : t\geq 2 \right\}, \\
  P_s &:= \{0\} \cup \left\{ t^{-s} : t\geq 1 \right\} \quad (\text{for } s>0), \\
  E_c &:= \{0\} \cup \left\{ c^{-t} : t\geq 0 \right\} \quad (\text{for } c>1),
\end{align*}
corresponding to logarithmic, polynomial (with power $s$), and exponential convergence. It is not difficult to show that $d_M(L) = d_M(\sqrt{L})  =1$, $d_M(P_s) = \frac{1}{s+1}$ and $d_M(\sqrt{P_s}) = \frac{2}{s+2}$, and $d_M(E_c) = d_M(\sqrt{E_c}) = 0$. 

Now, in a set with positive Minkowski dimension $d>0$, the optimal packings and coverings with distance $\epsilon$ scale as $\left({1}/{\epsilon}\right)^d$. This is relevant both in the converse and the direct part of our main Theorems \ref{thm:pessimistic_final} and \ref{thm:optimistic_final}, as $\epsilon=O\left(n^{-\frac12}\right)$ and $\epsilon=O\left(n^{-\alpha}\right)$, respectively, in terms of the block length $n$. Hence, optimal packings and coverings grow as a power of $n$, translating into the exponent $\Theta(n\log n)$ of $N(n,\lambda_1,\lambda_2)$ we have seen. 
Clearly, both $\Gamma_\epsilon(F)$ and $\Pi_\epsilon(F)$ always grow as a function  of $\frac{1}{\epsilon}$, however if the dimension is zero, it will be according to a different law slower than any power. For instance, for the exponential sequence set $E_c$ above, note that $\sqrt
{E_c} = E_{\sqrt{c}}$, and it is not hard to see that 
\[
  \log_c\frac{c-1}{3\epsilon} 
    \leq \Gamma_{\epsilon}(E_c) 
    \leq \log_c\frac{c^2}{2\epsilon}.
\]
The upper bound results from placing an interval of length $2\epsilon$ at $[0;2\epsilon]$, which covers all $t$ such that $c^{-t}\leq 2\epsilon$, plus a separate interval for each smaller $t$. The lower bound comes from the realization that if the gap between two consecutive points in $E_c$ is bigger than $2\epsilon$, then they cannot be covered by the same interval. Say, $2\epsilon < 3\epsilon \leq c^{1-t}-c^{-t} = (c-1)c^{-t}$, which is true for $t\leq \log\frac{c-1}{3\epsilon}$. This means that each such $c^{-t}$ needs its own interval. 

Going through the proof of Theorem \ref{thm:pessimistic_final} with these tight bounds for covering 
numbers as the cardinalities of the covering sets, shows for this channel $B\vert_{E_c}$ that
\begin{equation}
  \label{eq:DI-for-E-channel}
  2^{n\log\log n - O(n)}
   \leq N_{\text{DI}}(n,\lambda_1,\lambda_2) 
   \leq 2^{n\log\log n + O(n)}.  
\end{equation}
In other words, here we have a channel where the maximum message length $\log N_{\text{DI}}(n,\lambda_1,\lambda_2)$ has the scale $n\log\log n$, and the suitably scaled DI capacity is $1$, irrespective of $c>1$.

We can rip off the proof of Theorem \ref{thm:pessimistic_final} with the general but abstract terms $\Pi_\epsilon(\widetilde{\cX})$ and $\Gamma_\delta(\widetilde{\cX})$ for the packing and covering numbers. This gives the following. 
\begin{theorem}
\label{thm:Abstract_Generalization}
Given a channel $W:\cX\rightarrow\cY$, $\lambda_1,\lambda_2 > 0$ with $\lambda_1+\lambda_2 < 1$ and $\alpha<\frac14$, there exist $\delta>0$ and $K$ such that
\begin{align*}
\log N_{\text{DI}}(n,\lambda_1,\lambda_2)
    &\leq n\log\Pi_{n^{-\alpha}} \!\left(\!\sqrt{\!\widetilde{\cX}}\right) 
\,\,\,\forall n\text{ sufficiently large}, \\
\log N_{\text{DI}}(n,\lambda_1,\lambda_2)
    &\geq n\log \Gamma_{\frac{\delta}{\sqrt{n}}} \!\left(\!\sqrt{\!\widetilde{\cX}}\right)+Kn
\,\,\, \forall n.
    \hspace{1.45cm}\qedsymbol
\end{align*}
\end{theorem}
This theorem meaningfully generalises all the previous capacity results in deterministic identification over channels with finite output. Notice that if we have a positive Minkowski dimension we can immediately recover the optimistic and pessimistic superexponential capacity bounds in Theorems \ref{thm:pessimistic_final} and \ref{thm:optimistic_final}. For the DMC, the covering and packing numbers are equal to the number of distinct rows of the stochastic matrix that describes the channel. Therefore, defining an exponential capacity we recover the result in Equation~\eqref{eq:Capacity_DMC}. We can also find capacity bounds for the special examples presented in this section just by defining a capacity in the correct scaling, as we have done in Equation~\eqref{eq:DI-for-E-channel}. Studying these abstract covering and packing numbers given a channel we can find meaningful bounds to suitably defined capacities, in the sense that we have information in the rates no matter the scaling.

\subsection{Superactivation of the DI capacity}
\label{subsec:superactivation}

Superactivation of channels -- a phenomenon where two zero-capacity channels combine to create a channel with positive capacity -- has been studied in quantum information theory since the highly celebrated discovery of superactivation for entanglement transmission over quantum channels \cite{SY:0CapacityChannels,SSY:0capacityGaussian}. This had been preceded by the discovery of superactivation of distillable entanglement from bipartite quantum states assuming the existence of NPT bound entangled states \cite{VollbrechtWolf:PPT-NPT} and the unconditional such demonstration for multipartite states \cite{SST:superactivation}. No classical analogue of this effect had been known before and for another decade after, so that some even expressed the belief that it was a unique feature of quantum information and could not appear in classical systems (as for instance in Renato Renner's 2012 ICMP address).
Indeed, superactivation cannot occur for the Shannon transmission capacity and a number of other classical channel capacities, but it was eventually shown to appear in classical settings for secrecy capacities. Note however that due to the additivity of the Csisz\'ar-K\"orner private capacity formula, this does not occur for the stationary and memoryless wiretap channels \cite{CK:wiretap}. The first example to be found was the secure transmission of messages through arbitrary varying wiretap channels (AVWCs) \cite{BS:superact_wiretap,SBP:SecureComsAdversial}, which are highly non-memoryless (see \cite{NWB:AVwiretap} for a full characterisation). After that, superactivation was also proven for secure identification over the AVWC \cite{BD:secureIDwiretap}. Zero-error transmission \cite{Shannon:zero-error}, too, exhibits superadditivity, but it cannot achieve superactivation, only superadditivity of the capacity between channels with positive capacities \cite{Haemers:theta,Alon:zero-error}. 

Since then, it has remained an open problem to find a classical communication task over classical channels such that superactivation occurs for the corresponding communication capacity -- in an i.i.d.~setting and not for a secrecy capacity. Here we address this question by showing that DI over restricted Bernoulli channels can easily exhibit superactivation. We hold that deterministic identification is indeed a relevant communication task, and crucially the channels involved are not only stationary and memoryless but also very simple (compared, say, to the quantum channels exhibiting superactivation). 

Everything follows from the general results in the preceding subsections and the study of the lower Minkowski dimension of families of compact sets. The following property is observed:
\begin{proposition}[{cf.~\cite[Prop.~3.4]{Robinson:dimensions}}]
\label{prop:dimension}
Let $F$ and $G$ be two compact sets with lower Minkowski dimensions $\underline{d}_M(F)$ and $\underline{d}_M(G)$, respectively. Then the lower Minkowski dimension of their Cartesian product is superadditive: $\underline{d}_M(F\times G) \geq \underline{d}_M(F)+\underline{d}_M(G)$, and the inequality can be strict in general. 
\end{proposition}

We present a particular example for the latter part of the statement, taken from \cite[Thm.~5.11]{Falconer85:fractal2} and \cite{Edgar:book} (see also \cite{Robinson:dimensions,RobinsonSharples}),
where $\underline{d}_M(F) = \underline{d}_M(G) = 0$ and $\underline{d}_M(F\times G)\geq 1$. 
Let $0=m_0<m_1<m_2<\dots$ be a sufficiently rapidly increasing sequence of integers. We define $F$ as the subset of real numbers in the unit interval $[0;1]$ whose decimal expansion has a $0$ in the $r$-th place for all $r$ with $m_k < r \leq m_ {k+1}$ and $k$ an even number. $G$ is defined similarly as the subset of real numbers in $[0;1]$ whose decimal expansion has a $0$ in the $r$-th place for all $r$ with $m_k < r \leq m_{k+1}$, but with $k$ an odd number.

For $k$ even, if we look at the first $m_{k+1}$ decimal places, we see that we can cover $F$ by $10^{j_k}$ intervals of length $\epsilon_k=10^{-m_{k+1}}$, where
\[
  j_k=(m_2-m_1)+(m_4-m_3)+\dots+(m_k-m_{k-1}).
\]
Then, the lower Minkowski dimension is given by
\[
  \underline{d}_M(F) 
  \leq \liminf_{k\rightarrow\infty} \frac{\log 10^{j_k}}{-\log 10^{-m_{k+1}}} = 0,
\]
if the sequence $m_k$ increases rapidly enough (once numbers up to $m_k$ are chosen, $j_k$ is defined, so we just need to choose $m_{k+1}$ large enough such that the covering rate is smaller than, say, $1/k$). With similar considerations for odd $k$ values we also find $\underline{d}_M(G)=0$. Now, we notice that by construction the sum of the sets is $F+G=[0;1]$, the whole unit interval, and as the addition function on the real numbers $\RR\times\RR\rightarrow\RR$ is Lipschitz, we get
\[
  \underline{d}_M(F\times G)
    \geq \underline{d}_M(F + G) 
    =    \underline{d}_M([0;1]) =1,
\]
because Lipschitz functions cannot increase the dimension.

\begin{theorem}
\label{thm:superactivation}
The superexponential deterministic identification capacity $\dot{C}_{\text{DI}}(W_1)$ and $\dot{C}_{\text{DI}}(W_2)$ of two channels $W_1$ and $W_2$, respectively, is superadditive:
\[
  \dot{C}_{\text{DI}}(W_1\times W_2)
    \geq\dot{C}_{\text{DI}}(W_1)+\dot{C}_{\text{DI}}(W_2).
\]
The inequality can be strict, indeed it can be so even if $\dot{C}_{\text{DI}}(W_1)=\dot{C}_{\text{DI}}(W_2)=0$, meaning that the channels can superactivate each other: $\dot{C}_{\text{DI}}(W_1\times W_2)>0$.
\end{theorem}
\begin{proof}
The first part is immediate, and follows from the definition of the DI capacity. Note that it is important here that we are dealing with the pessimistic capacity, requiring guaranteed rates for every sufficiently large block length, so that for the product channel we can consider the Cartesian product of codes for all sufficiently large block lengths. 

To prove the superactivation part we take the two sets $F,G \subset [0;1]$ defined above, with $\underline{d}_M(F) = \underline{d}_M(G) = 0$ and $\underline{d}_M(F\times G) = 1$, and define $W_1 := B\vert_{F}$, $W_2 := B\vert_{G}$, the Bernoulli channel restricted to the respective subsets. By Theorem \ref{thm:pessimistic_final}, $\dot{C}_{\text{DI}}(W_1) \leq \frac12 \underline{d}_M(F) = 0$ and likewise $\dot{C}_{\text{DI}}(W_2) \leq \frac12 \underline{d}_M(G) = 0$. On the other hand, $W_1 \times W_2 = (B\times B)\vert_{F\times G}$ is the restriction of $B\times B$ to $F\times G$. Again by Theorem \ref{thm:pessimistic_final}, we have $\dot{C}_{\text{DI}}(W_1 \times W_2) \geq \frac14 \underline{d}_M(F\times G) = \frac14 > 0$.
\end{proof}

\subsection[{A continuous and compact example}]{A continuous and compact example with additive noise}
\label{subsec:example}
Moving outside the universe of finite output channels, let us look at the following example of a continuous-input continuous-output channel (or rather family of channels), where the alphabets are at least compact manifolds and the transmission capacity is finite without the need for power constraints. All channels in the family have $\cX = \cY = \RR/\ZZ$, the additive group of real numbers modulo integers. It is naturally isomorphic to $U(1) = \{ z=e^{2\pi i x} : x\in[0;1]\}$, the multiplicative group of complex units. Let $0<\theta<1$ 
and define $A = A^{(\theta)}:\cX\rightarrow\cY$ by
\begin{equation}
  \label{eq:additive-arc-channel}
  A : x \mapsto A_x = U_{[x;x+\theta]\!\!\!\!\mod\ZZ},
\end{equation}
where the interval $[x;x+\theta]$ is wrapped around modulo integers, and $U_{[x;x+\theta]\!\!\!\!\mod\ZZ}$ is the uniform distribution on that set. This is an additive noise channel since one can describe the joint distribution of input $X$ and output $Y$ by stating $Y = X + N$, where $N$ is a random variable independent of $X$ with uniform distribution on $[0;\theta]$ (modulo $\ZZ$).

It is straightforward to see that the ordinary (single exponential) Shannon capacity of this channel is 
\[
  C(A^{(\theta)}) = -\log\theta,
\] 
the maximum mutual information, due to the covariance of the channel, achieved by the uniform input distribution (the Lebesgue measure $\lambda$ on $\RR/\ZZ \simeq [0;1]$). Note already here that all channel output distributions $A_x$ are absolutely continuous with respect to $\lambda$, indeed
\begin{equation}
  \label{eq:density-A_x}
  \frac{{\rm d}A_x}{{\rm d}\lambda}(y) 
   = \begin{cases}
       \frac{1}{\theta} & \text{ for } y\in [x;x+\theta]\!\!\!\!\mod\ZZ, \\
       0                & \text{ otherwise.}
     \end{cases}
\end{equation}

On the other hand, we can use the insights from the preceding subsections to shed light on the DI capacity of $A=A^{(\theta)}$. 
To start with, the set $\widetilde{\cX} = A(\cX) \subset \cP(\cY)$ is quite clearly one-dimensional, the map $x\mapsto A_x$ being Lipschitz continuous between the unit interval modulo integers (with the usual metric) and the probability distributions (with the total variation distance): 
\begin{equation}
  \label{eq:d_A}
  \myfrac{\frac12\|A_x-A_{x'}\|_1}
  {:=d_A(x,x')}
    = \begin{cases}
        \frac{\delta}{\theta} 
        & \text{for } 
        0\leq \delta \leq\min\{\theta,1-\theta\}, \\
        1 
        & \text{for } 
         \phantom{1-;}\,
         \theta\leq\delta\leq\frac12
         \quad (\theta\leq\frac12), \\
        \frac{1}{\theta}-1 
        & \text{for } 
         1-\theta\leq\delta\leq\frac12
         \quad (\theta\geq\frac12),
      \end{cases}
\end{equation}
where $\delta=\min\{|x-x'|,1-|x-x'|\}$ (and where we have silently chosen to represent each equivalence class in $\RR/\ZZ$ by a number $x,x' \in [0;1)$).

We shall show that the superexponential DI capacity of every channel $A^{(\theta)}$ is $1$, more precisely we have the following.

\begin{theorem}
For any $0<\theta<1$, and for all $\lambda_1,\lambda_2>0$ with $\lambda_1+\lambda_2<1$, 
\begin{equation*}
  \dot{C}_{\text{DI}}\left(A^{(\theta)}\right) = \lim_{n\rightarrow\infty} \frac{1}{n\log n}\log N_{\text{DI}}(n,\lambda_1,\lambda_2) = 1. 
\end{equation*}
\end{theorem}

\begin{proof}
For the strong converse we imitate the converse argument from Theorem \ref{thm:pessimistic_final} (note that as $\cY$ is not discrete, we cannot simply invoke that theorem). 
For an $(n,N,\lambda_1,\lambda_2)$-DI code $\{(u_j,\cE_j) : j=1,\ldots,N\}$, as we have already observed repeatedly, 
any two distinct code words $u_j=x^n$ and $u_k = {x'}^n$ ($j\neq k$) are separated: 
\[\begin{split}
  0 < \eta := 1-\lambda_1-\lambda_2 
     &\leq \frac12 \left\| A_{x^n} - A_{{x'}^n} \right\|_1\\
     &\leq \sum_{i=1}^n \frac12 \left\| A_{x_i} - A_{x_i'} \right\|_1. 
\end{split}\]
Now choose a set $\cX_0 \subset \RR/\ZZ$ of $K = \lceil \frac{2n}{\eta\theta} \rceil$ equispaced points, thus forming an $\frac{\eta\theta}{4n}$-net in the natural metric of $\RR/\ZZ$, and an $\frac{\eta}{4n}$-net with respect to $d_A$ [Equation \eqref{eq:d_A}]. If we replace each code word $x^n$ by the (letterwise) nearest net point $\xi^n$, by the triangle inequality the $\xi^n$ and ${\xi'}^n$ of distinct messages are still separated  (and hence pairwise distinct points on $\cX_0^n$): 
\[
  \sum_{i=1}^n \frac12 \left\| A_{\xi_i} - A_{\xi_i'} \right\|_1 \geq \eta/2 > 0.
\]
Indeed, these code words with the same decoding sets form an $(n,N,\lambda_1+\eta/4,\lambda_2+\eta/4)$-DI code.
Thus, $\log N \leq n\log K \leq n\log n +O(n)$. 

For the direct part, we cannot use Theorem \ref{thm:pessimistic_final} either, nor indeed Lemma \ref{lemma:abstract-hypo-testing} on which it rests, as both rely on finite output alphabet due to typicality. But we do not have to, since for a sequence $x^n\in\cX^n$, the support of $A_{x^n}$ does just fine. Namely, 
\[
  \cS_{x^n} := \operatorname{supp} A_{x^n} 
   = \prod_{i=1}^n \bigl([x_i;x_i+\theta]\!\!\!\!\mod\ZZ\bigr)
\]
is already the ideal hypothesis test to discriminate between $A_{x^n}$ and any $A_{{x'}^n}$:
\[
  1-\frac12\|A_{x^n}-A_{{x'}^n}\|_1 
    = A_{{x'}^n}(\cS_{x^n})
    = \prod_{i=1}^n \bigl(1-d_A(x_i,x_i')\bigr).
\]
We continue now imitating the direct part of the proof of Theorem \ref{thm:pessimistic_final}: given a small constant $t>0$ and any block length $n$, choose $k=\left\lfloor nt^2\right\rfloor$ and let $\cX_0 = \left\{\frac{x}{k}+\ZZ : x=0,1,\ldots,k-1\right\}$. We also pick a maximum code $\cC_t\subset\cX_0^n$ of minimum distance $>tn$. Then for $n$ sufficiently big so that $\frac1k \leq \min\{\theta,1-\theta\}$, any two distinct code words $u_j=x^n \neq u_k={x'}^n \in \cC_t$ differ in more than $tn$ positions, and in any one of them, say $i\in[n]$, $1-d_A(x_i,x_i') \leq 1-\frac{1}{k\theta} \leq 1-\frac{1}{nt^2\theta}$. By the above formulas, 
\[
  A_{u_k}(\cS_{u_j}) 
   \leq \left(1-\frac{1}{nt^2\theta}\right)^{nt} 
   \leq e^{-{1}/{t\theta}},
\]
which is the error probability of second kind $\lambda_2$, 
while obviously $A_{u_j}(\cS_{u_j})=1$, making the error probability of first kind $\lambda_1=0$. 

Clearly, we can make $\lambda_2$ arbitrarily small by choosing $t>0$ small enough, and the code size, by the same Gilbert-Varshamov argument as before, is as large as
\[
  N = |\cC_t| 
    \geq 2^{-n} k^{n(1-t)} 
    \geq \exp\bigl((1-t)n\log n - O(n)\bigr), 
\]
for sufficiently large $n$. As $t>0$ can be made arbitrarily small, this concludes the proof.
\end{proof}

This is the first example of a superexponential DI capacity with a precisely determined value, rather than being bounded from below and above. 
The example is a bit singular, in that the output distributions $A_x$ have non-full support, so we avoid talking about typicality by taking the support $\cS_{x^n}$ of $A_{x^n}$ as a natural decoding set. 
At first sight it is a little mysterious why the factor $\frac12$ from the upper bound of Theorem \ref{thm:pessimistic_final} does not show up here. But there is an explanation of this due to the infinite dimension of the Hilbert space $L^2(\cY,\lambda)$ in which the functions $\alpha_x(y) = \sqrt{\frac{{\rm d}A_x}{{\rm d}\lambda}(y)}$ from Equation \eqref{eq:density-A_x} live. Namely, a direct calculation shows that 
\[
  F(A_x,A_{x'}) 
    = \int_0^1 \lambda({\rm d}y) \alpha_x(y)\alpha_{x'}(y)
    = 1 - \frac12\|A_x-A_{x'}\|_1, 
\]
implying via Equations \eqref{eq:trace-distance-lower-spherization} and \eqref{eq:trace-distance-upper-spherization} that 
\[
  \frac12\|A_x-A_{x'}\|_1 
    \leq \left| \alpha_x - \alpha_{x'} \right|_2^2
    \leq 2\|A_x-A_{x'}\|_1.  
\]
So, as we expected from our experience with the spherisations in finite dimension, the $O(1/n)$-net required on the distributions $A_x$ translates into an $O(1/\sqrt{n})$-net on the $\alpha_x$. But while the geometry of the $A_x$ is essentially that of the circle of the parameters $x \in \RR/\ZZ$, the curve of the $\alpha_x \in L^2(\cY,\lambda)$ meanders through this infinite dimensional Hilbert space, which is why we cannot leverage the more favourable power of $n$ in the packing. 

However more importantly, once the candidate decoding set is fixed, we can directly bound the error probability of second kind, rather than going through the trace distance first and then losing performance in having to go back to $\cS_{x^n}$, suggesting possible strategies to improve the lower bound for other channels, or even in the general case.

\section[{DI via quantum channels}]{Deterministic identification\protect\\ via quantum channels}
\label{sec:q-deterministic_ID}
Identification via quantum channels was first introduced by L\"ober \cite{loeber:PhDThesis}, following Ahlswede and Dueck's model of randomized encoders \cite{AD:ID_ViaChannels} described in the introduction. This line of study has been developed significantly since then \cite{AW:StrongConverse,Winter:QCidentification,Winter:QC-ID-randomness,Winter:Review,HaydenWinter:QC-ID,TSW:soft-covering,CDBW:ID-0-sim:ICC,CDBW:ID-0-sim}. 
A major theme in these prior works has been the distinction between \emph{simultaneous} and \emph{non-simultaneous} decoders, and more recently between pure state and general mixed state encodings, as an attempt to introduce deterministic identification through quantum channels \cite{CDBW:ID-0-sim:ICC,CDBW:ID-0-sim}. However, all these models achieve doubly exponential code sizes, albeit with different associated capacities, so they do not seem to capture the essence of deterministic identification, in particular not the exponential or slightly superexponential scaling of the code size observed before.

\subsection{Quantum information preliminaries}
Before discussing the quantum generalisation of deterministic identification codes, we briefly review the quantum statistical formalism needed to do information theory, cf.~\cite{holevo:stat-structure-book,Wilde-book,QCN}. Quantum systems are described by complex Hilbert spaces, denoted $A$, $B$, etc, which we shall always assume to be of finite dimension $|A|$, $|B|$, etc. (the notation reflecting the cardinality of a basis of the Hilbert space, which coincides with the number of quantum degrees of freedom associated with it). Composite systems are formed by tensor products $A\ox B$, etc. 
Quantum states are described by density matrices $\rho$, which are positive semidefinite, $\rho\geq 0$, and of unit trace, $\Tr\rho=1$. The set of states on a given system $A$ is denoted $\cS(A)$. For a state $\rho^{AB} \in\cS(A\ox B)$, the marginals are $\rho^A = \Tr_B\rho^{AB} \in \cS(A)$ and $\rho^B = \Tr_A\rho^{AB} \in \cS(B)$. 

Information about quantum states is accessed through measurements, which are generally given by positive operator-valued measures (POVMs): a POVM is a family $(M_u:u\in\cU)$ of positive semidefinite operators $M_u\geq 0$ such that $\sum_u M_u = \1$. The fundamental relation between states and measurements is expressed in Born's rule,
\[
  \Pr\{u|\rho\} = \Tr\rho M_u, 
\]
predicting the probability of the POVM finding the element $M_u$ corresponding to $u$ when the system is prepared in state $\rho$.
Probability distributions and classical statistics are embedded into this formalism by way of diagonal matrices: given a distinguished (``classical'') orthonormal basis $\{\ket{x}:x\in\cX\}$ of a Hilbert space $X$, a probability distribution $p=(p_x)_x\in\cP(\cX)$ can be represented by the density matrix $\sum_x p_x \proj{x}$, and a POVM $(M_u)_u$ consisting of matrices diagonal in the distinguished basis, $M_u = \sum_x M(u|x)\proj{x}$, can be interpreted as classical response functions defining a random observation $U$.

Quantum channels are completely positive trace-preserving linear maps denoted by $\cN:A\rightarrow B$, even though the map is from $\cL(A)$ to $\cL(B)$, the set of matrices in $A$ and $B$, respectively. To be a physical map, $\cN$ has to be positive and trace-preserving, meaning that it maps quantum states in $A$ denoted by $\cS(A)$ to $\cS(B)$, and the same has to hold for any channel extension $\cN\ox\id_C$. 

A particularly important class of cptp maps is that of \emph{classical-quantum (cq-)channels} $\cN:X\rightarrow B$, with a distinguished orthonormal basis $\{\ket{x}:x\in\cX\}$. It has the general form $\cN(\ket{x}\bra{x'}) = \delta_{xx'} W_x$, with states $W_x\in\cS(B)$. Operationally, this can be described as a channel that first measures in the computational basis (the POVM consisting of the projectors $\proj{x}$) and then prepares the state $W_x$. Therefore, a cq-channel (both mathematically and physically) is given by the map 
\[\begin{split}
   W: \cX&\rightarrow\cS(B)\\
        x&\mapsto W_x.
\end{split}\]
While here we will consider cptp maps only between finite-dimensional Hilbert spaces, we allow the extension of this latter definition to all measurable spaces $\cX$, along with the cq-channel $W:\cX\rightarrow\cS(B)$ a measurable map with respect to the Borel sigma-algebra on $\cS(B)$. The $n$-extension $W^n:\cX^n\rightarrow\cS(B^n)$ maps input words $x^n\in\cX^n$ to $W_{x^n}=W_{x_1}\otimes\cdots\otimes W_{x_n}$.

Analogously to Definition \ref{def:transmission_code}, an $(n,M,\lambda)$-transmission code for the memoryless channel $\cN$, i.e. over $\cN^{\ox n}:A^n=A^{\ox n}\rightarrow B^n=B^{\ox n}$, is a collection $\{(\pi_m,D_m):m\in[M]\}$ of states $\pi_m\in\cS(A^n)$ and positive semidefinite operators $D_m\geq 0$ on $B^n$ such that $\sum_m D_m \leq \1$ (forming a sub-POVM), and such that for all $m\in[M]$, $\Tr\cN^{\ox n}(\pi_m) D_m \geq 1-\lambda$. The maximum $M$ is denoted $M(n,\lambda)$ as before, but if all the $\pi_m$ are tensor product states (or more generally separable states) across the $n$ systems $A^n$, we denote the maximum $M_{\text{sep}}(n,\lambda)$. 

For a cq-channel $W:\cX\rightarrow\cS(B)$, we consider a code to consist of code words $u_m\in\cX^n$ rather than quantum states, but otherwise unchanged. Note that in this case, the encodings are by definition separable. The Holevo-Schumacher-Westmoreland (HSW) theorem and subsequent refinements describe the capacity of a quantum channel. 

\begin{theorem}[\cite{holevo:capacity,SW:capacity,Winter:Strong_converse,ON:Strong_converse}]
\label{thm:HSW}
The following formula gives the transmission capacity of a memoryless cq-channel $W$, and the strong converse holds, namely for all $\lambda\in (0;1)$, 
\[
  C(W) = \lim_{n\rightarrow\infty}\frac{1}{n}\log M(n,\lambda)
       = \max_{P\in\cP(\cX)} I(P;W).
\]
Here, $I(P;W)=H(PW)-H(W|P)$ is the Holevo information with the von Neumann entropy $H(\rho)=-\Tr\rho\log\rho$, and we are using the notation
$PW =\displaystyle{\int_\cX} P({\rm d}x)W_x \in \cS(B)$, and the conditional entropy $H(W|P)= \displaystyle{\int_\cX} P({\rm d}x) H(W(\cdot|x))$.

More generally, for a quantum channel $\cN:A\rightarrow B$ and separable encodings, 
\[\begin{split}
  C_{\text{sep}}(\cN) 
   :&= \inf_{\lambda>0}\liminf_{n\rightarrow\infty}
    \frac1n \log M_{\text{sep}}(n,\lambda) \\
    &= \lim_{n\rightarrow\infty}
    \frac1n \log M_{\text{sep}}(n,\lambda) 
    = \chi(\cN),
\end{split}\]
with the \emph{Holevo capacity}
\[
  \chi(\cN) = \max_{\{p_x,\pi_x\}} \left[H\left(\sum_x p_x\cN(\pi_x)\right) - \sum_x p_x H\left(\cN(\pi_x)\right)\right]. 
\]

Instead, for general (entangled) encodings, the ultimate classical capacity is
\[\begin{split}
  C(\cN) 
   :&= \inf_{\lambda>0}\liminf_{n\rightarrow\infty}
    \frac1n \log M(n,\lambda) \\
    &= \inf_{\lambda>0}\limsup_{n\rightarrow\infty}
    \frac1n \log M(n,\lambda) 
    = \sup_n \frac1n \chi(\cN^{\ox n}).
\end{split}
\]
\end{theorem}

Note that the part about separable encodings is actually a consequence of the case of cq-channels, because for given $\cN:A\rightarrow B$ we can define a cq-channel $W:\cX\rightarrow\cS(B)$ with $W_x=\cN(x)$ for $x\in\cX:=\cS(A)$, and every code for $\cN^{\ox n}$ with separable encodings is essentially equivalent to a code for $W^n$.

The statistical distance measures discussed for probability distributions generalise to quantum states as the \emph{trace distance} between quantum states (density matrices) and the \emph{fidelity}. The trace distance of two density matrices $\rho$ and $\sigma$ is 
\[
  \frac{1}{2}\|\rho-\sigma\|_1 
    :=\frac12\Tr\left|\rho-\sigma\right|
    = \frac12\Tr\sqrt{(\rho-\sigma)^2}.
\]
In particular, the trace distance between two diagonal density matrices (which correspond to classical probability distributions) is the total variation distance between the distributions.
The fidelity of quantum states on the other hand is 
\[\begin{split}
  F(\rho,\sigma)
   := \left\|\sqrt{\rho}\sqrt{\sigma}\right\|_1
   &= \Tr\sqrt{\sqrt{\rho}\sigma\sqrt{\rho}}\\
   &= \min_{(F_i)} \sum_i\sqrt{(\Tr\rho F_i)(\Tr\sigma F_i)},
\end{split}\]
where $(F_i)$ is an arbitrary POVM, i.e.~$F_i\geq 0$ and $\sum_i F_i=\1$. The fidelity and the trace distance of quantum states are related to each other through the Fuchs-van de Graaf inequalities \cite{FVdG:ineq}:
\begin{equation}
  \label{eq:FvdG}
  1-F(\rho,\sigma)
    \leq \frac{1}{2}\|\rho-\sigma\|_1
    \leq\sqrt{1-F(\rho,\sigma)^2} =: p(\rho,\sigma),
\end{equation}
the latter function known as the \emph{purified distance}.

Generalising Lemma \ref{lemma:purified-HS}, we can find a relation between the Hilbert-Schmidt and the trace distances which will be useful for the main theorem in this section (Theorem \ref{thm:cq-DI}). 
\begin{lemma}
\label{lemma:HS-Trace}
Let $\rho$ and $\sigma$ be two quantum states, then 
\[\begin{split}
  \frac12\|\rho-\sigma\|_1 
    \leq p(\rho,\sigma) 
    \leq \left\|\sqrt{\rho}-\sqrt{\sigma}\right\|_2
    &\leq \sqrt{2}\,p(\rho,\sigma)\\
    &\leq \sqrt{2}\sqrt{\|\rho-\sigma\|_1}.
\end{split}
\]
\end{lemma}
\begin{proof}
The leftmost and the rightmost inequalities are nothing but the Fuchs-van de Graaf relations \eqref{eq:FvdG}, so we may focus on the middle two inequalties. 
We start with the squared Hilbert-Schmidt distance:
\[\begin{split}
    \|{\sqrt{\rho}-\sqrt{\sigma}}\|_2^2
      &= \Tr \left(\sqrt{\rho}-\sqrt{\sigma}\right)^2\\
      &= \Tr \rho + \Tr \sigma - 2\Tr \sqrt{\rho}\sqrt{\sigma}\\
      &= 2\left(1-\Tr\sqrt{\rho}\sqrt{\sigma}\right).
\end{split}\]
To continue, we invoke the fidelity property $F(\rho,\sigma) \geq \Tr\sqrt{\rho}\sqrt{\sigma} \geq F(\rho,\sigma)^2$, proved easily from the definition of the fidelity and optical purifications according to Uhlmann's theorem, cf.~\cite{Holevo:Fidelity,Wilde:Fidelity}. 
Then, from the above we get on the one hand
\[\begin{split}
    \|{\sqrt{\rho}-\sqrt{\sigma}}\|_2^2
      &=    2\left(1-\Tr\sqrt{\rho}\sqrt{\sigma}\right) \\
      &\geq 2\left(1-F(\rho,\sigma)\right) \\
      &\geq 1-F(\rho,\sigma)^2 
       =    p(\rho,\sigma)^2. 
\end{split}\]
On the other hand, 
\[\begin{split}
    \|{\sqrt{\rho}-\sqrt{\sigma}}\|_2^2
      &=    2\left(1-\Tr\sqrt{\rho}\sqrt{\sigma}\right) \\
      &\leq 2\left(1-F(\rho,\sigma)^2\right) \\
      &=    2 p(\rho,\sigma)^2,
\end{split}\]
which concludes the proof.
\end{proof}

\subsection[{ID and DI via quantum channels}]{From randomised to deterministic identification via quantum channels}
To generalise (deterministic) identification to the quantum setting, we start with cq-channels $W:\cX\rightarrow\cS(B)$, with a finite-dimensional complex Hilbert space $B$ and a measurable space $\cX$. 
Since the input is classical, the following definitions follow directly the example of Definition \ref{def:ID_code} and its subsequent discussion: 

\begin{definition}[L\"ober \cite{loeber:PhDThesis}, Ahlswede/Winter \cite{AW:StrongConverse}]
\label{def:qID}
An $(n,N,\lambda_1,\lambda_2)$-\emph{identification code} for $W$ is a collection of $\{(P_j,E_j):j\in[N]\}$ of probability distributions $P_j$ on $\cX^n$ and POVM elements $0\leq E_j\leq\1$ on $B^n$, such that for all $j\neq k\in[N]$
\[
  \Tr(P_jW^n)E_j \geq 1-\lambda_1,
  \quad 
  \Tr(P_jW^n)E_k \leq \lambda_2,
\]
where $P_jW^n =\displaystyle{\int_{{\cX^n}}} P_j({\rm d}^n x^n)W_{x^n} \in \cS(B^n)$. 

If all the POVMs $(E_j,\1-E_j)$ are coexistent, we call the code \emph{simultaneous} \cite{loeber:PhDThesis}. This means that there exists a POVM $(D_m:m\in\cM)$ such that all $E_j$ are obtained by coarse-graining from it. In other words, there are subsets $\cM_j\subset\cM$ such that $E_j = \sum_{m\in\cM_j} D_m$. 
\end{definition}

The largest $N$ in this definition such that a good ID code (simultaneous ID code) exists is denoted $N(n,\lambda_1,\lambda_2)$ [$N^{\text{sim}}(n,\lambda_1,\lambda_2)$], and the corresponding double exponential capacities are defined as in the classical case: 
\begin{equation}\begin{split}
  \ddot{C}_\text{ID}(W)
  &:=\inf_{\lambda_1,\lambda_2>0} \liminf_{n\rightarrow\infty} \frac{1}{n} \log\log N(n,\lambda_1,\lambda_2), \\
  \ddot{C}_\text{ID}^{\text{sim}}(W)
  &:=\inf_{\lambda_1,\lambda_2>0} \liminf_{n\rightarrow\infty} \frac{1}{n} \log\log N^{\text{sim}}(n,\lambda_1,\lambda_2).
\end{split}\end{equation}
For cq-channels, these capacities have been determined to be equal, and to equal the (single exponential) transmission capacity of $W$: 

\begin{theorem}[L\"ober \cite{loeber:PhDThesis}, Ahlswede/Winter \cite{AW:StrongConverse}]
\label{thm:cq-ID}
For a cq-channel $W:\cX\rightarrow\cS(B)$, we have
\(
  \ddot{C}_\text{ID}(W) 
   = \ddot{C}_\text{ID}^\text{sim}(W) 
   = C(W),
\)
and indeed the strong converse holds for any $\lambda_1,\lambda_2>0$ with $\lambda_1+\lambda_2<1$:
\[\begin{split}
C(W)&=\lim_{n\rightarrow\infty} \frac1n \log\log N(n,\lambda_1,\lambda_2)\\
  &= \lim_{n\rightarrow\infty} \frac1n \log\log N^\text{sim}(n,\lambda_1,\lambda_2).
  \end{split}
  \]
\end{theorem}

For general quantum channels (where now to each message is associated an encoding state $\rho_j\in\cS(A^n)$, in general mixed), simultaneous and non-simultaneous ID capacities turn out to be different, although both have a double exponential capacity and $\ddot{C}_{\text{ID}}(\cN) \geq \ddot{C}_{\text{ID}}^{\text{sim}}(\cN) \geq C(\cN)$ \cite{loeber:PhDThesis}. 
For example, for the noiseless qubit channel $\id_2$, $\ddot{C}_{\text{ID}}(\id_2) = 2$ \cite{Winter:QC-ID-randomness}, whereas it was shown recently \cite{TSW:soft-covering,CDBW:ID-0-sim:ICC,CDBW:ID-0-sim} that $\ddot{C}_{\text{ID}}^{\text{sim}}(\id_2) = 1$, demonstrating that the first inequality is generally strict. It remains unknown whether the second is an equality or strict. 

Note that as in the transmission problem, we can also here consider the restriction that the encodings $\rho_j\in\cS(A^n)$ be separable states across the $n$ systems, giving rise to the maximum code sizes $N^{\text{sep}}(n,\lambda_1,\lambda_2)$ and $N^{\text{sep,sim}}(n,\lambda_1,\lambda_2)$, and the double exponential capacities $\ddot{C}_{\text{ID}}^{\text{sep}}(\cN)$ and $\ddot{C}_{\text{ID}}^{\text{sep,sim}}(\cN)$, respectively. Then the same reduction to the associated cq-channel $W_x=\cN(X)$ for $x\in\cX:=\cS(A)$ shows that in general, $\ddot{C}_{\text{ID}}^{\text{sep}}(\cN) = \ddot{C}_{\text{ID}}^{\text{sep,sim}}(\cN) = \chi(\cN)$. 

Following now the classical development, we call an ID code for the memoryless cq-channel $W$ \emph{deterministic} if all the distributions $P_j$ are point masses $\delta_{u_j}$ with $u_j\in\cX^n$. As before, $N_\text{DI}(n,\lambda_1,\lambda_2)$ and $N_\text{DI}^\text{sim}(n,\lambda_1,\lambda_2)$ are defined as the largest $N$ such that a deterministic $(n,N,\lambda_1,\lambda_2)$-ID code and a simultaneous and deterministic code exist respectively. This gives rise to the following DI capacities, defined at the slightly superexponential scale:
\begin{equation}\begin{split}
  \dot{C}_\text{DI}(W)
  &:=\inf_{\lambda_1,\lambda_2>0} \liminf_{n\rightarrow\infty} \frac{1}{n\log n} \log N_{\text{DI}}(n,\lambda_1,\lambda_2), \\
  \dot{C}_\text{DI}^{\text{sim}}(W)
  &:=\inf_{\lambda_1,\lambda_2>0} \liminf_{n\rightarrow\infty} \frac{1}{n\log n} \log N_{\text{DI}}^{\text{sim}}(n,\lambda_1,\lambda_2).
\end{split}\end{equation}

Now we can prove the following analogue, and indeed generalisation of Theorem \ref{thm:pessimistic_final}. In particular, we find that the maximum code size for DI codes is still of the order $2^{Rn\log n}$ if the Minkowski dimension of the output of the channel is positive. To state it, we need to fix an informationally complete POVM $T=(T_y:y\in\cY)$ on $B$ with $\infty > |\cY|\geq |B|^2$. Then we can form the concatenated channel $\overline{W}=\cT\circ W$ of $W$ followed by the quantum-classical (qc-) channel $\cT(\sigma) = \sum_y (\Tr\sigma T_y) \proj{y}$, so that we can identify $\overline{W}$ with the classical channel $\overline{W}(y|x)=\Tr W_xT_y$. Observe that $\widehat{\cX} = \overline{W}(\cX)$ is a linear one-to-one image of $\widetilde{\cX} = W(\cX)$, hence they share their upper and lower Minkowski dimensions.

\begin{theorem}
\label{thm:cq-DI}
The slightly superexponential simultaneous and general DI capacities of a cq-channel $W:\cX\rightarrow\cS(B)$ are bounded as follows (for any $\lambda_1,\lambda_2>0$ with $\lambda_1+\lambda_2<1$):
\[
\begin{split}
\frac14{\underline{d}_M\!\left(\sqrt{\!\widehat{\cX}}\right)}
   \leq \dot{C}_{\text{DI}}^\text{sim}(W) 
   &\leq \liminf_{n\rightarrow\infty} \frac{1}{n\log n} \log N_{\text{DI}}(n,\lambda_1,\lambda_2)\\
   &\leq\frac12 \underline{d}_M\!\left(\!\sqrt{\!\widetilde{\cX}}\right),
\end{split}
\]
where $\widehat{\cX}=\overline{W}(\cX)\subset\cP(\cY)$, $\widetilde{\cX}=W(\cX)\subset\cS(B)$ and $\sqrt{\!\widetilde{\cX}} = \left\{ \sqrt{\sigma} : \sigma\in\widetilde{\cX} \right\}$. 
\end{theorem}

\begin{proof}
We follow closely the classical case. For the strong converse part, we start with deterministic $(n,N,\lambda_1,\lambda_2)$-ID codes $\{(u_j,E_j) : j\in[N] \}$. We know that the trace distance between code words is lower bounded $\frac12\|W_{u_j}-W_{u_k}\|_1 \geq 1-\lambda_1,\lambda_2 =: 3r > 0$, hence by Lemma \ref{lemma:HS-Trace},
\[
  \|{\sqrt{W_{u_j}}-\sqrt{W_{u_k}}}\|_2
    \geq \frac12 \|{W_{u_j}-W_{u_k}}\|_1
    \geq 3r.
\]
Pick a sequence of $\delta_k > 0$ converging to $0$ such that 
\[
  \lim_{k\rightarrow\infty} \frac{\Gamma_{\delta_k}\!\left(\!\sqrt{\!\widetilde{\cX}}\right)}{-\log\delta_k} 
  = \underline{d}_M\!\left(\!\sqrt{\!\widetilde{\cX}}\right) =: d, 
\]
and for each $k$ choose a minimal $\delta_k$-covering of $\sqrt{\!\widetilde{\cX}}$. Defining the block lengths $n_k=\lfloor r^2/\delta_k^2\rfloor$, we automatically have $\frac{r}{\sqrt{n_k}}$-coverings with Hilbert-Schmidt balls centered at points from a subset $\cX_{r\!/\!\sqrt{n_k}}\subset\cX$ such that for every $x\in\cX$ there is a $\zeta\in\cX_{r\!/\!\sqrt{n_k}}$ with $\|{\sqrt{W_x}-\sqrt{W_\zeta}}\|_2\leq r/\sqrt{n_k}$. By definition of the Minkowski dimension, for every $\epsilon > 0$ we can find a $K > 0$ such that $\lvert{\cX_{r\!/\!\sqrt{n_k}}}\rvert \leq \left(K\sqrt{n_k}/r\right)^{d+\epsilon}$. Now, by replacing the code words $u_j = x^{n_k}$ in our code by $u_j' = \zeta^{n_k}$, defined as the letterwise nearest word in $\cX_{r\!/\!\sqrt{n_k}}^{n_k}$, we find another $(n,N,\lambda_1+r,\lambda_2+r)$-DI code, which thus necessarily has $N\leq\lvert{\cX_{r\!/\!\sqrt{n_k}}}\rvert^{n_k}$, so taking the logarithm we notice 
\(\log N\leq\frac12(d+\epsilon)n_k\log n_k+O(n_k)\), and therefore:
\[
\begin{split}
\dot{C}_{\text{DI}}(W)
  &\leq \liminf_{n\rightarrow\infty} \frac{\log N_{\text{DI}}(n,\lambda_1,\lambda_2)}{n \log n}\\
  &\leq \liminf_{k\rightarrow\infty} \frac{\log N_{\text{DI}}(n_k,\lambda_1,\lambda_2)}{n_k \log n_k}\leq\frac{1}{2}\underline{d}_M\!\left(\!\sqrt{\!\widetilde{\cX}}\right).
\end{split}
\]

For the direct part, we reduce the statement directly to Theorem \ref{thm:pessimistic_final}, since $\overline{W}$ is a classical channel and $\widehat{\cX} \subset \cP(\cY)$ is its image of output distributions. Note that the resulting code for $W$ is automatically simultaneous because we measure the output quantum states with the fixed POVM $T^{\ox n}$, followed by classical post-processing. 
\end{proof}

Finally, we can return to the question of how to approach deterministic identification over general quantum channels $\cN:A\rightarrow B$. We have hinted already at the beginning of this section that if we allow arbitrary pure state encodings, motivated by the idea that pure states are the least randomised quantum states, this still results in doubly exponential code sizes, with our without simultaneity of the decoder imposed \cite{CDBW:ID-0-sim:ICC,CDBW:ID-0-sim}. 
Another, more stringent characteristic of classical deterministic identification is of course that the encodings are words over the input alphabet; in quantum language that would mean that they are tensor products, 
$\rho_j = \rho_j^{(1)}\ox\cdots\ox\rho_j^{(n)}$ with $\rho_j^{(i)}\in\cS(A)$. 

By fixing additionally a subset $\cX\subset\cS(A)$, we propose to define an \emph{$(n,N,\lambda_1,\lambda_2)$-deterministic identification (DI) $\cX$-code} for the pair $(\cN,\cX)$ as an $(n,N,\lambda_1,\lambda_2)$-ID code 
\[
  \left\{ (\rho_j= \rho_j^{(1)}\ox\cdots\ox\rho_j^{(n)},E_j) : j\in[N] \right\} \quad \text{with} \quad \rho_j^{(i)}\in\cX.
\]
For the purpose of the following result, the maximum number of messages in a code is denoted $N_{\text{DI},\cX}(n,\lambda_1,\lambda_2)$ and $N_{\text{DI},\cX}^{\text{sim}}(n,\lambda_1,\lambda_2)$ in the general and simultaneous case respectively. The capacities are defined as usual in the superexponential scale, and we may denote them $\dot{C}_{\text{DI}}(\cN\vert_{\cX})$ and $\dot{C}_{\text{DI}}^\text{sim}(\cN\vert_{\cX})$. 

Thanks to the reduction, already repeated twice, of the quantum channel $\cN$ with input restriction $\cX\subset\cS(A)$ to a cq-channel $W:\cX\rightarrow\cS(B)$, $W_x=\cN(x)$, which maps deterministic identification $\cX$-codes for $\cN$ into deterministic identification codes for $W$ and vice versa, we immediately get the following corollary.

\begin{corollary}
\label{cor:q-DI}
For a quantum channel $\cN:A\rightarrow B$ and a product state restriction to encodings in $\cX\subset\cS(A)$, the slightly superexponential simultaneous and general DI capacities are bounded as follows:
\[
\begin{split}
\frac14{\underline{d}_M\!\left(\!\sqrt{\!\widehat{\cX}}\right)}
   \leq \dot{C}_{\text{DI}}^\text{sim}(\cN\vert_{\cX}) 
   &\leq \liminf_{n\rightarrow\infty} \frac{\log N_{\text{DI},\cX}(n,\lambda_1,\lambda_2)}{n\log n} \\
   &\leq \frac12\underline{d}_M\!\left(\!\sqrt{\!\widetilde{\cX}}\right),
\end{split}  
\]
where $\widetilde{\cX} = \cN(\cX) \subset\cS(B)$ and $\widehat{\cX} = \overline{W}(\cX)$ as before. 
\qed
\end{corollary}

Evidently, we can also emulate the analysis of the optimistic deterministic identification capacity (simultaneous or general) for cq-channels $W:\cX\rightarrow \cS(B)$. By adapting the proof of Theorem \ref{thm:optimistic_final} along the lines of the the proof of Theorem \ref{thm:cq-DI}, we obtain the following: 

\begin{theorem}
\label{thm:cq-DI-optimistic}
The slightly superexponential simultaneous and general optimistic DI capacities of a cq-channel $W:\cX\rightarrow\cS(B)$ are bounded as follows (for any $\lambda_1,\lambda_2>0$ with $\lambda_1+\lambda_2<1$):
\[
\begin{split}
  \frac14{\overline{d}_M\! \left(\!\sqrt{\!\widehat{\cX}}\right)}
   \leq \dot{C}_{\text{DI}}^\text{sim,opt}(W) 
   &\leq \limsup_{n\rightarrow\infty} \frac{\log N_{\text{DI}}(n,\lambda_1,\lambda_2)}{n\log n}\\ 
   &\leq\frac12 \overline{d}_M\!\left(\!\sqrt{\!\widetilde{\cX}}\right).
  \hspace{2.4cm}\qedsymbol
\end{split}
\]
\end{theorem}

\section{Discussion}
\label{sec:discussion}

\subsection{A tale of two capacities}
In 1968, Ahlswede showed that communication over non-stationary (yet memoryless) channels can easily have pessimistic capacity $\underline{C}=0$ and optimistic capacity $\overline{C}>0$ \cite{Ahlswede1968}. Through such cases he argued that, despite the optimistic and pessimistic capacity numbers containing important information about the channel, they are far from giving the complete picture. Ahlswede instead suggested to use the entire sequence of rates $R(n) = \frac1n \log N$ for every block length $n$, for which he showed a coding theorem and strong converse. In 2006, he returned to the issue in \cite{Ahlswede2006}, addressing the ``lack of precision in terminology in Information Theory''. There, the concepts of pessimistic and optimistic capacities appeared again formally defined (cf.~\cite{CK:book2011}), along with different examples where either the optimistic or pessimistic capacities clearly give an incomplete picture of the particular communication problem at hand. He qualified the optimistic view as a ``dream world'' for multi-way and compound channels, and the pessimistic view as ``an Information Theoretic Perpetuum Mobile'', indeed pointing out that saying a channel has no capacity because $\underline{C}=0$ is not reasonable. He concluded from this that there is no ``true'' capacity of a channel, and highlighted the importance of correctly choosing and clearly defining the capacity we use in our problems.

The distinction between these different capacity definitions is normally hidden behind the fact that in most simple cases (for stationary DMCs and average DMCs, for example) the optimistic and pessimistic capacities coincide, and a gap between these two only appeared in much more complex (non-stationary) models. However, in the present paper we show a separation in an i.i.d.~setting for deterministic identification capacities, prompting us to revisit Ahlswede's discussion.

The gap between the optimistic and pessimistic capacities is also highly relevant to understand the present superactivation results for the DI capacity. Notice that in all the scenarios where we observe superactivation in classical settings, both the secure transmission and identification over the non-memoryless AVWC and the particular DI scenarios discussed in this paper (which are i.i.d.), we are exploiting the gap between optimistic and pessimistic capacities. Indeed, let us define $W_1$ and $W_2$ as two channel instances, and $\underline{C}$ and $\overline{C}$ as the pessimistic and optimistic capacities respectively on any of the previous cases. Superactivation amounts to $\underline{C}(W_1)=\underline{C}(W_2)=0$ and $\underline{C}(W_1\times W_2)>0$. In all those known classical cases mentioned, we can observe that additionally $\overline{C}(W_1)>0$, $\overline{C}(W_2)>0$, and $\overline{C}(W_1) + \overline{C}(W_2)= \overline{C}(W_1\times W_2) \geq \underline{C}(W_1\times W_2)$. In other words, superactivation of the pessimistic capacity is accompanied by a gap between the pessimistic and optimistic capacity, and the latter even turns out to be additive. 

In the present paper, when comparing Theorems \ref{thm:pessimistic_final} (pessimistic DI capacity $\dot{C}_{\text{DI}}$) and \ref{thm:optimistic_final} (optimistic DI capacity $\dot{C}_{\text{DI}}^{\text{opt}}$) in the light of the superactivation example for the former in Theorem \ref{thm:superactivation}, we can come to similar conclusions using known relations between the lower and upper Minkowski dimensions of sets and their Cartesian products. The following results is a refinement of Proposition \ref{prop:dimension}: 

\begin{proposition}[{Robinson/Sharples~\cite{RobinsonSharples}}]
\label{prop:dimension-improved}
For any two subsets $F$ and $G$ in a metric space, 
\[\begin{split}
  \underline{d}_M(F)+\underline{d}_M(G) 
    &\leq \underline{d}_M(F\times G) \\
    &\leq \left\{ \myfrac{\underline{d}_M(F)+\overline{d}_M(G)}{\overline{d}_M(F)+\underline{d}_M(G)} \right\}\\ 
    &\leq \overline{d}_M(F\times G)
    \leq \overline{d}_M(F)+\overline{d}_M(G),
\end{split}\]
and all of the inequalities can be strict. 
\end{proposition}

In any example of superactivation of the DI capacity, i.e.~when $\dot{C}_{\text{DI}}(W_1) = \dot{C}_{\text{DI}}(W_2) = 0$ but $\dot{C}_{\text{DI}}(W_1\times W_2) > 0$, our results show that this is equivalent to $\underline{d}_M(F) = \underline{d}_M(G) = 0$ for $F=\sqrt{\!\widetilde{\cX}_1}$, $G=\sqrt{\!\widetilde{\cX}_2}$, and $\underline{d}_M(F\times G) > 0$. Now the second inequality of Proposition \ref{prop:dimension-improved} implies that then necessarily $\overline{d}_M(F)>0$ and $\overline{d}_M(G)>0$, meaning that $\dot{C}_{\text{DI}}^{\text{opt}}(W_1),\, \dot{C}_{\text{DI}}^{\text{opt}}(W_2) > 0$. 
What this means is that even though both $W_1$ and $W_2$ have infinitely many block lengths where each one performs poorly, superactivation is only possible if both of them equally have infinitely many block lengths where each one achieves a positive (slightly superexponential) rate. Indeed, this is how the example of Theorem \ref{thm:superactivation} is constructed: for a long stretch of block lengths, $W_1$ is bad but $W_2$ makes up for it by being good at those, followed by a long stretch of block lengths where  $W_2$ is bad but $W_1$ makes up for it by being good, and so on alternatingly. 

Ahlswede \cite{Ahlswede1968,Ahlswede2006} derived his conclusions from memoryless channels with strong time variation, but should we heed his cautionary tale also in the present case of DI capacity over i.i.d.~channels? Here, the variability comes from the fact that increasing block lengths correspond to the successive zooming in to the fractal geometry of $\widetilde{\cX}$ and that can produce a certain nonuniformity. It remains the common feature that superactivation goes along with positive optimistic capacity, which could be interpreted as a certain ``potential'' for capacity (cf.~\cite{SSW:symmetric,YangWinter:potential}). Indeed, this leads to new questions: first, regarding the converse of the above observation, does a channel with zero lower Minkowski dimension but positive upper Minkowski dimension of $\sqrt{\!\widetilde{\cX}}$ activate or even superactivate some other channel with regards to their DI capacity? Secondly, if on the other hand the Minkowski dimension of $\sqrt{\!\widetilde{\cX}}$ is well-defined (i.e.,~lower and upper variant coincide), are then the pessimistic and the optimistic capacity equal? More to the point, is the DI capacity additive for two such channels?

\subsection{Conclusions and further open questions}
By considering general finite-output but arbitrary-input memoryless channels, we have shown that the superexponential scaling with an exponent of order $n\log n$ in the block length is a general feature of deterministic identification codes via noisy channels. Furthermore, the optimal (redefined) rate is related, through upper and lower bounds, to the Minkowski dimension of the set of output probability distribution inside the probability simplex over the output alphabet. This is surprising since capacities are more commonly related to metric aspects of this output set. For instance, Shannon's communication capacity of the same channel is given by the divergence radius. In contrast, our DI capacity formulas (or rather bounds) are intrinsically scale invariant. 
We furthermore provide the very general Theorem \ref{thm:Abstract_Generalization} that upper- and lower-bounds the maximum size of a DI code for general sets and block lengths, hence providing meaningful bounds for suitably defined capacities even when the Minkowski dimension is $0$.

We used these results to show that superactivation is possible in the communication setting of deterministic identification over classical memoryless channels, a first in the classical world, and have analysed the effect in terms of the behaviour of the lower and upper Minkowski dimension under Cartesian products, relating it to Ahlswede's analysis of the concepts of optimistic and pessimistic capacity and their insufficiency to fully capture the communication capabilities of a channel. 

Because of the insensitivity to metric, our results carry over almost unchanged to classical-quantum channels with finite-dimensional output system and to general quantum channels with a product state restriction on the encoding: all of these channels are essentially given by a subset of quantum state space at the output, one-half the lower (upper) Minkowski dimension of its square root set bounding from above the pessimistic (optimistic) deterministic identification capacity. In the quantum case, however, the lower bound does not match the upper bound as beautifully as in the classical case: we first have to map the output states through a fixed measurement to output probability distributions, so that one-quarter the Minkowski dimension of its square root set is a lower bound on the simultaneous DI capacity. Of course, this disparity might be due to the lower and upper bounds concerning different DI capacities, namely the simultaneous and general flavour, which we do not know to be equal or different (note, though, that with the restriction to tensor product inputs, the simultaneous and non-simultaneous randomized ID capacities coincide). However, the lower bound is dissatisfying also because it depends on the informationally complete measurement $T$ chosen in Theorem \ref{thm:cq-DI}; we would like to know how to express the result of optimising over all POVMs. Furthermore, is there a lower bound on the general DI capacity of a cq-channel that refers directly to the quantum states? The bottleneck towards a lower bound matching the converse seems to be a quantum state analogue of the Hypothesis Testing Lemma \ref{lemma:abstract-hypo-testing} for a suitably defined entropy typical projector, as all other elements for quantum are in place to carry out the quantum generalisation. 

We believe that our results go some way towards explaining the previous findings 
regarding Gaussian and Poisson channels. Indeed, the analysis of the Bernoulli channel implies the previous lower bounds on the DI capacities of those channels. 
The biggest open question, as in these prior works, is the determination of the exact superexponential capacity. Indeed, our results might be taken to suggest that $\dot{C}_{\text{DI}}(W) = \gamma\, \underline{d}_M\bigl(W(\cX)\bigr)$ with a universal constant $\gamma \in \left[\frac14;\frac12\right]$. Furthermore, the continuous example in Subsection \ref{subsec:example} where upper and lower bound coincide, and the apparent tightness of our converses (especially our improvement of the Poisson channel upper bound in Subsection \ref{subsec:poisson}, see also Remark \ref{Remark:direct_part}) lead us to speculate that the lower bounds could be the ones amenable to improvement, and even that the upper bound might be attainable in the general case ($\gamma=\frac12$). 
But to determine this remains for future investigation.

\section*{Acknowledgments}
The authors thank Alan Sheretz and Adam Beckenbaugh for invaluable insights into erroneous message identification under various constraints, dating back to several discussions with the last author in Benson, AZ.

\renewcommand*{\bibfont}{\footnotesize}
\printbibliography

\phantom{.}

\addtocontents{toc}{\protect\setcounter{tocdepth}{0}}

\begin{IEEEbiographynophoto}{Pau Colomer} (Student Member, IEEE) received the Graduate degree in physics from Universitat Aut\`onoma de Barcelona (UAB), Bellaterra, Spain, in 2021; and the Master's degree from Universitat de Barcelona (UB), UAB, and Universitat Polit\`ecnica de Catalunya (UPC), Barcelona, Spain, in 2022.

During 2022, he was part of Institut de Ci\`encies Fot\`oniques (ICFO), Castelldefels, Spain, and Grup d'Informaci\'o Qu\`antica (GIQ), Bellaterra, Spain, as an associate student researcher. Since 2023 he has been pursuing a Ph.D. at the Technical University of Munich (TUM), Munich, Germany.
\end{IEEEbiographynophoto}

\begin{IEEEbiographynophoto}{Christian Deppe} (Member, IEEE) received his Dipl.-Math. degree in mathematics from the University of Bielefeld in 1996 and his Dr.-math. degree, also from the University of Bielefeld, in 1998. He then worked there until 2010 as a research associate and assistant at the Faculty of Mathematics, Bielefeld. In 2011, he took over the management of the project "Safety and Robustness of the Quantum Repeater" from the Federal Ministry of Education and Research at the Faculty of Mathematics, Bielefeld University, for two years. In 2014, Christian Deppe was funded by a DFG project at the Chair of Theoretical Information Technology, Technical University of Munich. At the Friedrich Schiller University in Jena, Christian Deppe took up a temporary professorship at the Faculty of Mathematics and Computer Science in 2015. Until 2023, he worked for six years at the Chair of Communications Engineering at the Technical University of Munich and since January 2024 has taken on new tasks at the Institute of Communications Engineering at the TU Braunschweig. He is project leader of several projects funded by the BMBF and the DFG.
\end{IEEEbiographynophoto}

\begin{IEEEbiographynophoto}{Holger Boche}
(Fellow, IEEE) received the Dipl.-Ing. degree in electrical engineering, the Graduate degree in mathematics, and the Dr.-Ing. degree in electrical engineering from the Technische Universität Dresden, Dresden, Germany, in 1990, 1992, and 1994, respectively, the master’s degree from Friedrich-Schiller Universität Jena, Jena, Germany, in 1997, and the Dr.Rer.Nat. degree in pure mathematics from Technische Universität Berlin, Berlin, Germany, in 1998. 
 
In 1997, he joined the Fraunhofer Institute for Telecommunications, Heinrich-Hertz-Institute (HHI), Berlin. From 2002 to 2010, he was a Full Professor of mobile communication networks with the Institute for Communications Systems, Technische Universität Berlin. In 2003, he became the Director of the Fraunhofer German-Sino Laboratory for Mobile Communications, Berlin. In 2004, he became the Director of the Fraunhofer Institute for Telecommunications, HHI. He was a Visiting Professor with ETH Zürich, Zürich, Switzerland, from 2004 to 2006 (Winter), and with KTH Stockholm, Stockholm, Sweden, in 2005 (Summer). He joined the Institute of Theoretical Information Technology, Technical University of Munich (TUM), Munich, Germany, in October 2010, where he is currently a Full Professor. He has been a member and an Honorary Fellow of the TUM Institute for Advanced Study, Munich, since 2014. Since 2018, he has been the Founding Director of the Center for Quantum Engineering, TUM. Since 2021, he has been jointly leading the BMBF Research Hub 6G-Life with Frank Fitzek. Among his publications is the recent book, \textit{Information Theoretic Security and Privacy
of Information Systems} (Cambridge University Press, 2017).

Prof. Boche was elected as a member of the German Academy of Sciences (Leopoldina) in 2008 and the Berlin Brandenburg Academy of Sciences and Humanities in 2009. He is a member of the IEEE Signal Processing Society SPCOM and SPTM Technical Committees. He was a recipient of the Research Award “Technische Kommunikation” from the Alcatel SEL Foundation in October 2003, the “Innovation Award” from the Vodafone Foundation in June 2006, and the Gottfried Wilhelm Leibniz Prize from Deutsche Forschungsgemeinschaft (German Research Foundation) in 2008. He was a recipient of the 2007 IEEE Signal Processing Society Best Paper Award. He was a co-recipient of the 2006 IEEE Signal Processing Society Best Paper Award. He was the General Chair of the Symposium on Information Theoretic Approaches to Security and Privacy at IEEE GlobalSIP 2016.
\end{IEEEbiographynophoto}

\begin{IEEEbiographynophoto}{Andreas Winter}
received a Diploma degree in Mathematics from Freie Universit\"at Berlin, Germany, in 1997, and a Ph.D. degree from Fakult\"at f\"ur Mathematik, Universit\"at Bielefeld, Germany, in 1999.
He was Research Associate at the University of Bielefeld until 2001, and then with the Department of Computer Science at the University of Bristol, UK. In 2003, still with the University of Bristol, he was appointed Lecturer in Mathematics, and in 2006 Professor of Physics of Information. From 2007 to 2012 he was in addition a Visiting Research Professor with the Centre of Quantum Technologies at NUS, Singapore. 
Since 2012 he has been ICREA Research Professor with the Universitat Aut\`onoma de Barcelona, Spain. His research interests include quantum and classical Shannon theory, and discrete mathematics.

He is recipient, along with Charles H. Bennett, Igor Devetak, Aram W. Harrow and Peter W. Shor, of the 2017 Information Theory Society Paper Award. In 2022, he received an Alexander von Humboldt Research Prize, a Hans Fischer Senior Fellowship of Technische Universit\"at M\"unchen, 
and one of three 2022 QCMC International Quantum Awards.
\end{IEEEbiographynophoto}

\end{document}